\documentclass[oneside,english]{amsart}
\usepackage[T1]{fontenc}
\usepackage[latin9]{inputenc}
\usepackage{geometry}
\usepackage{verbatim}
\usepackage{amsthm}
\usepackage{amstext}
\usepackage{amssymb}
\usepackage{graphicx}
\usepackage{subfigure}
\usepackage{esint}
\usepackage{url}
\usepackage{yhmath}
\usepackage[all]{xy}
\usepackage{color}

\usepackage[percent]{overpic}

\usepackage[space]{grffile} 
\makeatletter
\numberwithin{equation}{section}
\numberwithin{figure}{section}
\theoremstyle{plain}
\newtheorem*{thm*}{\protect\theoremname}
\theoremstyle{plain}
\newtheorem{thm}{\protect\theoremname}

\theoremstyle{plain}
\newtheorem{lem}{\protect\lemmaname}
\theoremstyle{remark}
\newtheorem{rem}[lem]{\protect\remarkname}
\theoremstyle{plain}
\newtheorem{prop}[lem]{\protect\propositionname}
\theoremstyle{plain}
\newtheorem{cor}[lem]{\protect\corollaryname}
\theoremstyle{plain}

\theoremstyle{plain}

\theoremstyle{plain}

\theoremstyle{plain}

\theoremstyle{plain}

\theoremstyle{plain}

\usepackage{mathrsfs}
\usepackage{subfigure}

\newcommand{\SLE}{\mathrm{SLE}}
\newcommand{\SLEk}{\mathrm{SLE}_{\kappa}}

\newcommand{\Primary}{\Psi}
\newcommand{\PrimaryBlock}[2]{\Psi_{#1}^{#2}}
\newcommand{\sigmaseq}{\bar{\sigma}}

\newcommand{\sF}{\mathcal{F}}
\newcommand{\sZ}{\mathcal{Z}}

\newcommand{\sS}{\mathcal{S}}

\newcommand{\bR}{\mathbb{R}}
\newcommand{\R}{\bR}

\newcommand{\bZ}{\mathbb{Z}}

\newcommand{\bN}{\mathbb{N}}

\newcommand{\bZpos}{\mathbb{Z}_{> 0}}
\newcommand{\bZnn}{\mathbb{Z}_{\geq 0}}
\newcommand{\Zpos}{\bZpos}
\newcommand{\Znn}{\bZnn}
\newcommand{\bQ}{\mathbb{Q}}

\newcommand{\bC}{\mathbb{C}}

\newcommand{\ii}{\mathfrak{i}}

\newcommand{\Mob}{\mu}
\newcommand{\confmap}{\phi}

\newcommand{\pder}[1]{\frac{\partial}{\partial#1}}
\newcommand{\pdder}[1]{\frac{\partial^{2}}{\partial#1^{2}}}

\newcommand{\set}[1]{\left\{  #1\right\}  }

\newcommand{\slLie}{\mathfrak{sl}}
\newcommand{\Uqsltwo}{{U}_{q}(\mathfrak{sl}_{2})}
\newcommand{\Hcp}{\Delta}
\newcommand{\qnum}[1]{\left[#1\right] }

\newcommand{\Wd}{\mathsf{M}}
\newcommand{\HWsp}{\mathsf{H}}

\newcommand{\Wbas}{\mathrm{e}}

\newcommand{\Sbas}{\mathrm{s}}
\newcommand{\TRbas}{\mathrm{t}}

\newcommand{\Qrep}{\mathcal{Q}}

\newcommand{\dualbrakets}[1]{\left\langle #1 \right\rangle}
\newcommand{\dualpairing}[2]{\dualbrakets{ #1 , \, #2 }}

\newcommand{\Coblobastwodim}{\mathrm{u}}

\newcommand{\Catalan}{\mathrm{C}}

\newcommand{\KWleq}{\stackrel{\scriptscriptstyle{()}}{\scriptstyle{\longleftarrow}}} 
\newcommand{\Mmat}{\mathscr{M}}

\newcommand{\genMmat}{\mathfrak{M}}
\newcommand{\genMinv}{\mathfrak{M}^{-1}}

\newcommand{\dmn}{\mathrm{dim}}

\newcommand{\tens}{\otimes}

\newcommand{\id}{\mathrm{id}}
\newcommand{\isom}{\cong}

\newcommand{\chamber}{\mathfrak{X}}
\newcommand{\PartF}{\sZ}
\newcommand{\Sol}{\sS}

\newcommand{\ConfBlockFun}{\mathcal{U}}

\newcommand{\hwvec}{\mathrm{w}}

\newcommand{\twoFone}{ {{}_2 F_1} }

\newcommand{\DP}{\mathrm{DP}}
\newcommand{\DPleq}{\preceq} 

\newcommand{\wedgeat}[1]{\lozenge_#1} 
\newcommand{\upwedgeat}[1]{\wedge^#1}
\newcommand{\downwedgeat}[1]{\vee_#1}
\newcommand{\upwedge}{\wedge}
\newcommand{\downwedge}{\vee}
\newcommand{\upslope}{\nearrow}
\newcommand{\downslope}{\searrow}
\newcommand{\removewedge}[1]{\setminus \wedgeat{#1}}
\newcommand{\removeupwedge}[1]{\setminus \upwedgeat{#1}}
\newcommand{\removedownwedge}[1]{\setminus \downwedgeat{#1}}

\newcommand{\slopeat}[1]{\times_#1} 

\newcommand{\nestedtilingof}{T_0}

\newcommand{\CItilingsof}{\mathcal{C}}

\newcommand{\walk}{\alpha}
\newcommand{\emptywalk}{{(0)}}
\newcommand{\nodef}{s}

\newcommand{\NormalizationConstant}{c}

\makeatletter
\newcommand*{\centerfloat}{%
  \parindent \z@
  \leftskip \z@ \@plus 1fil \@minus \textwidth
  \rightskip\leftskip
  \parfillskip \z@skip}
\makeatother

\AtBeginDocument{

}

\makeatother

\usepackage{babel}
\providecommand{\corollaryname}{Corollary}
\providecommand{\lemmaname}{Lemma}
\providecommand{\propositionname}{Proposition}
\providecommand{\remarkname}{Remark}
\providecommand{\theoremname}{Theorem}
\providecommand{\conjecturename}{Conjecture}

\definecolor{kallecol}{rgb}{.75,.0,.55}

\definecolor{allucol}{rgb}{.3,.6,.2}

\definecolor{blue}{rgb}{0,0,1}

\definecolor{red}{rgb}{1,0,0}

\definecolor{green}{rgb}{0,1,0}

\begin{document}


\author{A.~Karrila, K.~Kytölä, and E.~Peltola}

\

\vspace{2.5cm}

\begin{center}
\LARGE \bf \scshape {
Conformal blocks, $q$-combinatorics, \\ and quantum group symmetry
}
\end{center}

\vspace{0.75cm}

\begin{center}
{\large \scshape Alex Karrila}\\
{\footnotesize{\tt alex.karrila@aalto.fi}}\\
{\small{Department of Mathematics and Systems Analysis}}\\
{\small{P.O. Box 11100, FI-00076 Aalto University, Finland}}\bigskip{}
\\
{\large \scshape Kalle Kyt\"ol\"a}\\
{\footnotesize{\tt kalle.kytola@aalto.fi}}\\
{\small{Department of Mathematics and Systems Analysis}}\\
{\small{P.O. Box 11100, FI-00076 Aalto University, Finland}\\
\url{https://math.aalto.fi/~kkytola/}}\bigskip{}
\\
{\large \scshape Eveliina Peltola}\\
{\footnotesize{\tt eveliina.peltola@unige.ch}}\\
{\small{Section de Math\'{e}matiques, Universit\'{e} de Gen\`{e}ve,}}\\
{\small{2--4 rue du Li\`{e}vre, Case Postale 64, 1211 Gen\`{e}ve 4, Switzerland}}
\end{center}

\vspace{0.75cm}


\begin{center}
\begin{minipage}{0.85\textwidth} \footnotesize
{\scshape Abstract.}
In this article, we find a $q$-analogue for Fomin's formulas.
The original Fomin's formulas relate determinants of random walk
excursion kernels to loop-erased random walk partition functions,
and our formulas analogously relate conformal block functions of
conformal field theories to pure partition functions of multiple SLE
random curves.
We also provide a construction of the conformal block functions by a
method based on a quantum group, the $q$-deformation of $\slLie_2$.
The construction both highlights the representation theoretic origin of conformal
block functions and explains the appearance of $q$-combinatorial formulas.
\end{minipage}
\end{center}


\bigskip{}

%


\section{Introduction}

Conformal blocks are fundamental building blocks of correlation functions of conformal field theories.
In this article, we study the combinatorics of conformal block functions associated to 
the simplest non-trivial primary fields in conformal field theories (CFT).

Following the conventions in the literature about random conformally invariant curves of $\SLEk$ type, we parameterize
the central charge of the 
CFT via a parameter $\kappa > 0$, as
\begin{align} \label{eq: central charge parametrization}
c=\; & \frac{(3\kappa-8)(6-\kappa)}{2\kappa} . 
\end{align}
We assume $\kappa \in (0,8) \setminus \bQ$.
The primary fields whose conformal blocks we study are of conformal weight
\begin{align*}
h = \frac{6 - \kappa}{2 \kappa} .
\end{align*}
This is the first non-trivial conformal weight in the Kac table~\cite{Kac-ICM_proceedings_Helsinki}, and
fields of this type appear in particular as the boundary changing fields that
create the tip of an $\SLEk$ type curve~\cite{BB-SLE_martingales_and_Virasoro_algebra, BB-CFTs_of_SLEs,
FW-conformal_restriction_highest_weight_representations_and_SLE,
BBK-multiple_SLEs, Dubedat-commutation, Kytola-local_mgales, KM-GFF_and_CFT,
Dubedat-SLE_and_Virasoro_representations_localization}.

We cover some background on conformal blocks in CFT 
in Section~\ref{sec: conformal block functions background}.
For all other parts of the article, a few key properties
of conformal block functions can be taken as their definition.
Namely, the partial differential equations, M\"obius covariance, and
asymptotics given precisely in Section~\ref{sub: defining conformal block functions}
serve as their defining properties. 

\begin{figure}
\centerfloat
\includegraphics[width = 0.45 \textwidth]{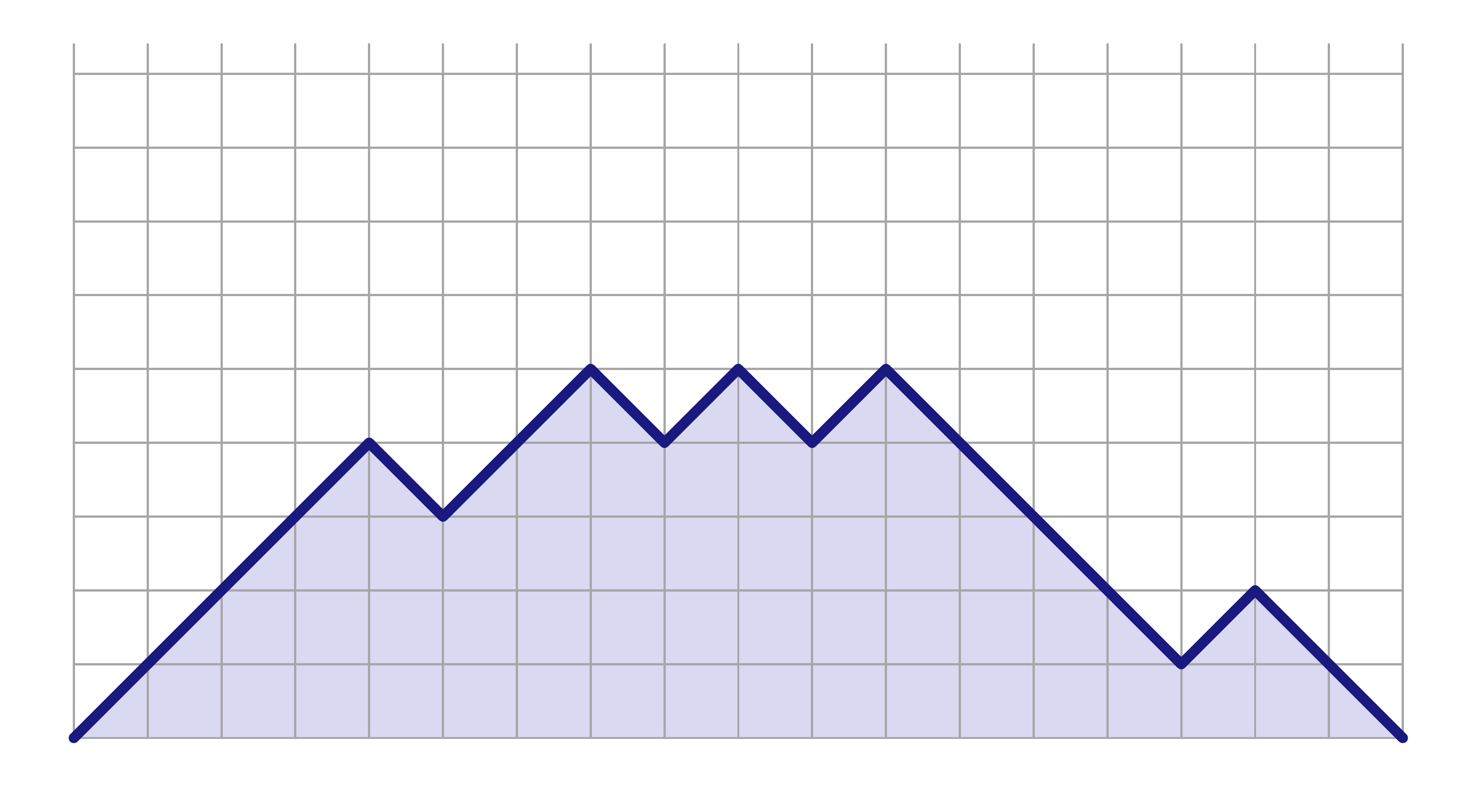} \qquad
\includegraphics[width = 0.45 \textwidth]{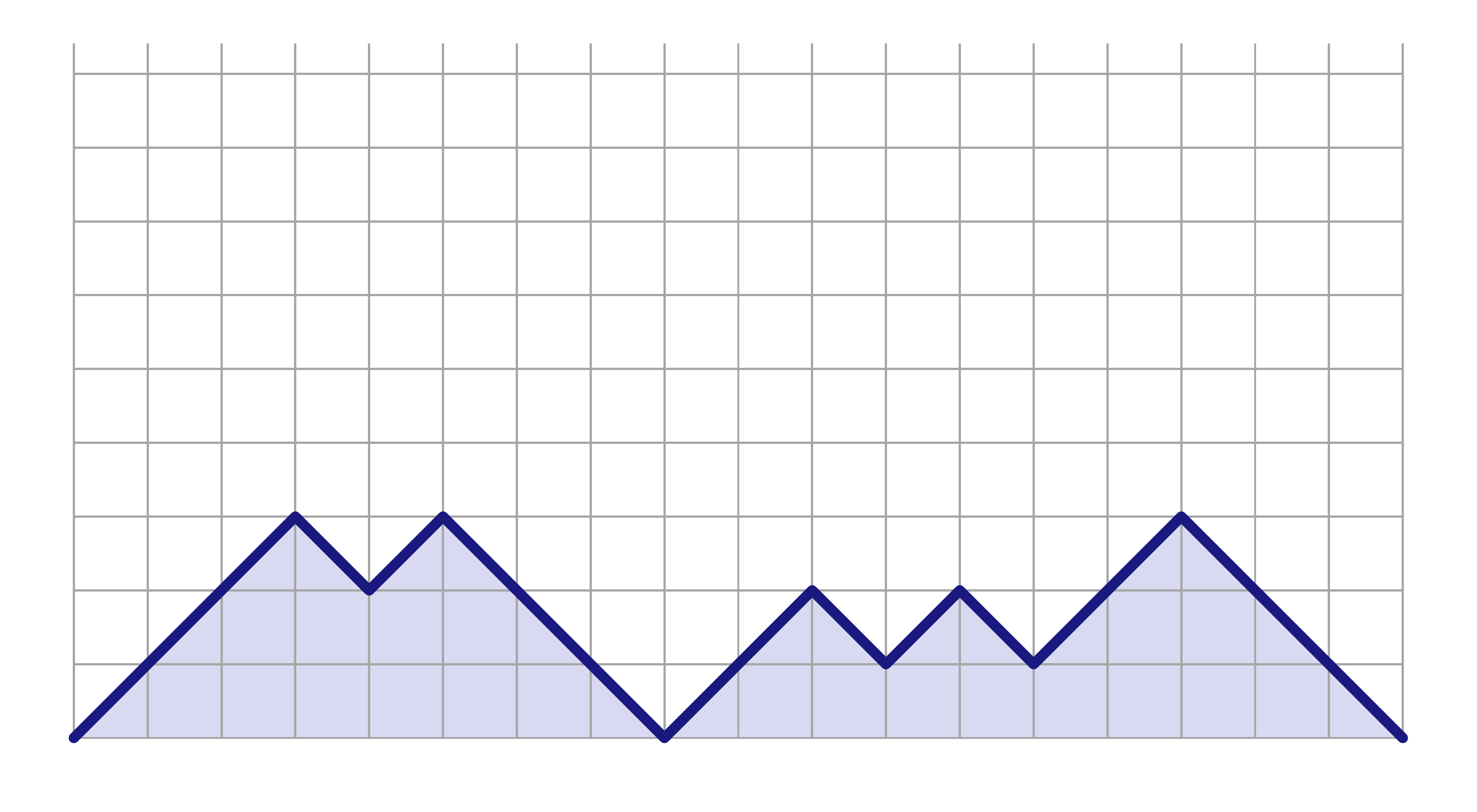} \\
\caption{\label{fig: Dyck paths}
Examples of Dyck paths.
}
\end{figure}
The starting point for the combinatorics
is the observation that the conformal block functions are functions
\begin{align*}
\ConfBlockFun_\alpha (x_1 , \ldots , x_{2N})
\end{align*}
of an even number $n = 2N$ of variables, which are
indexed by Dyck paths $\alpha$ of length $2N$, that is,
sequences $\alpha = (\alpha(0) , \alpha(1) , \ldots , \alpha(2N))$ of non-negative integers
with $|\alpha(j) - \alpha(j-1)| = 1$ for all $j$ and $\alpha(0) = \alpha(2N) = 0$.
Figure~\ref{fig: Dyck paths} depicts examples of Dyck paths.

Our first main result, Theorem~\ref{thm: change of basis theorem} given in Section~\ref{sec: change of basis}, 
relates the conformal block functions via explicit $q$-combinatorial formulas to
another family of functions: the pure partition functions of multiple $\SLE$s~\cite{KP-pure_partition_functions_of_multiple_SLEs},
whose precise definition we recall in Section~\ref{sub: pure partition functions}.
The pure partition functions are a key ingredient in the construction of joint laws of $N$
curves of $\SLEk$ type, with deterministic
connectivity~\cite{BBK-multiple_SLEs, KP-pure_partition_functions_of_multiple_SLEs,
KKP-boundary_correlations_in_planar_LERW_and_UST, PH-Global_multiple_SLEs_and_pure_partition_functions}.
They are indexed by the planar connectivities, or equivalently, 
by Dyck paths.
In the case $\kappa = 2$, 
a similar relation between the conformal block functions and the pure partition
functions arises as a consequence of Fomin's formulas~\cite{Fomin-LERW_and_total_positivity}
for loop-erased random walks, as explained in~\cite{KKP-boundary_correlations_in_planar_LERW_and_UST},
and our result can be seen as a $q$-analogue of Fomin's formulas.

Specifically, we show that for fixed $N$,
the conformal block functions and the multiple $\SLE$ pure partition functions 
form two bases of the same function space of dimension given by the $N$:th Catalan number 
$\Catalan_N = \frac{1}{N+1} \binom{2N}{N}$,
and we give an explicit combinatorial formula for the change of basis matrix $\genMmat$
from the latter basis to the former, as well as for the inverse $\genMinv$.
The rows and columns of both $\genMmat$ and $\genMinv$ are indexed by Dyck paths,
and the entries are rational functions of $q = e^{\ii 4 \pi / \kappa}$.
The non-zero entries of $\genMmat$ appear where a binary relation introduced in
\cite{KW-double_dimer_pairings_and_skew_Young_diagrams, SZ-path_representations_of_maximal_paraboloc_KL_polynomials}
holds between the two Dyck paths, whereas the
non-zero entries of $\genMinv$ appear where the two Dyck paths are in the
natural partial order. Combinatorial formulas for the matrices 
are
given in Section~\ref{sub: change of basis results}, but for small values of $N$
their forms are already illustrated in
Figures~\ref{fig: quantum KW matrices only} and~\ref{fig: quantum CIDT matrices only}.


\begin{figure}
\centerfloat
\includegraphics[width = 0.336 \textwidth]{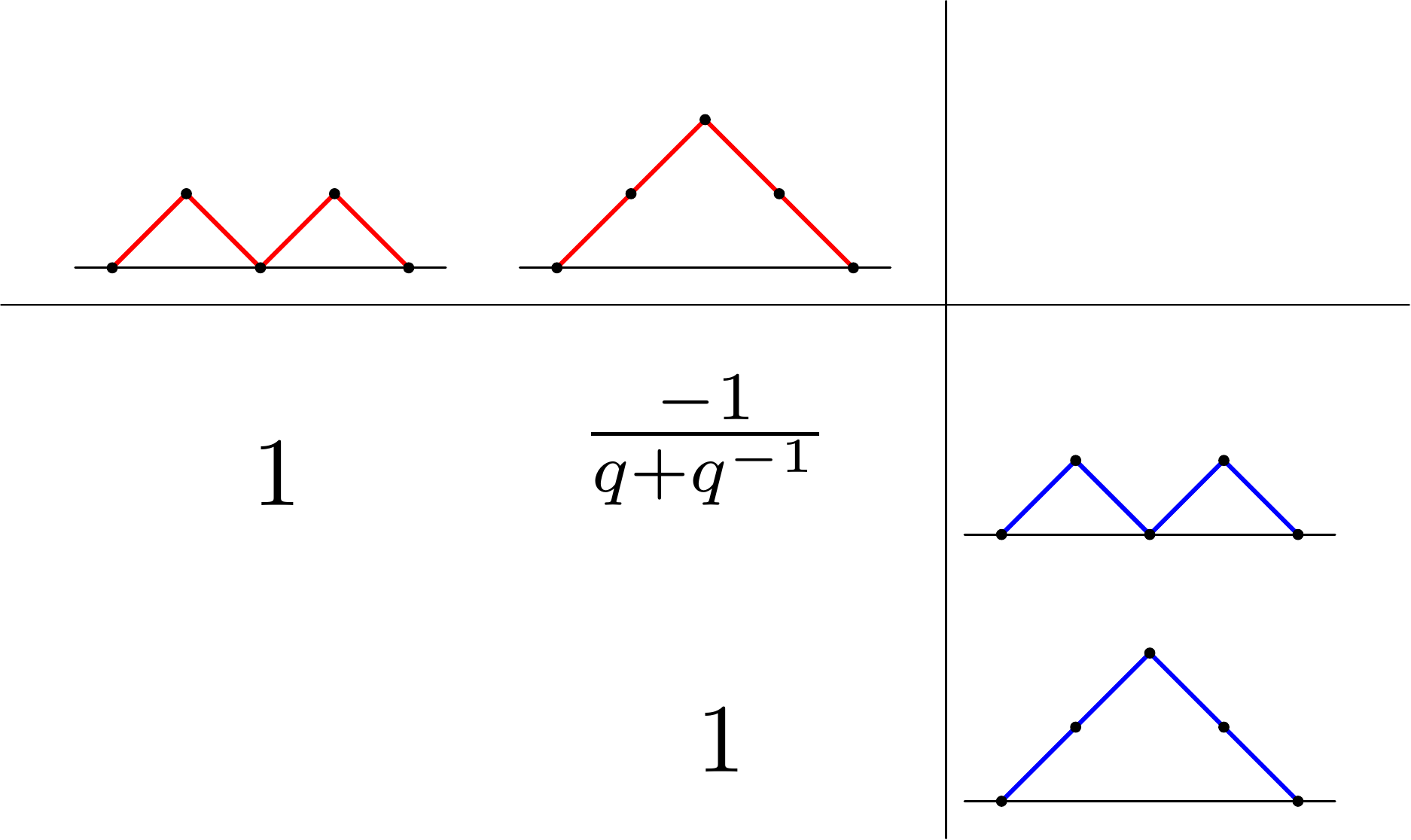} 
\hspace{0.5 cm}
\includegraphics[width = 0.54\textwidth]{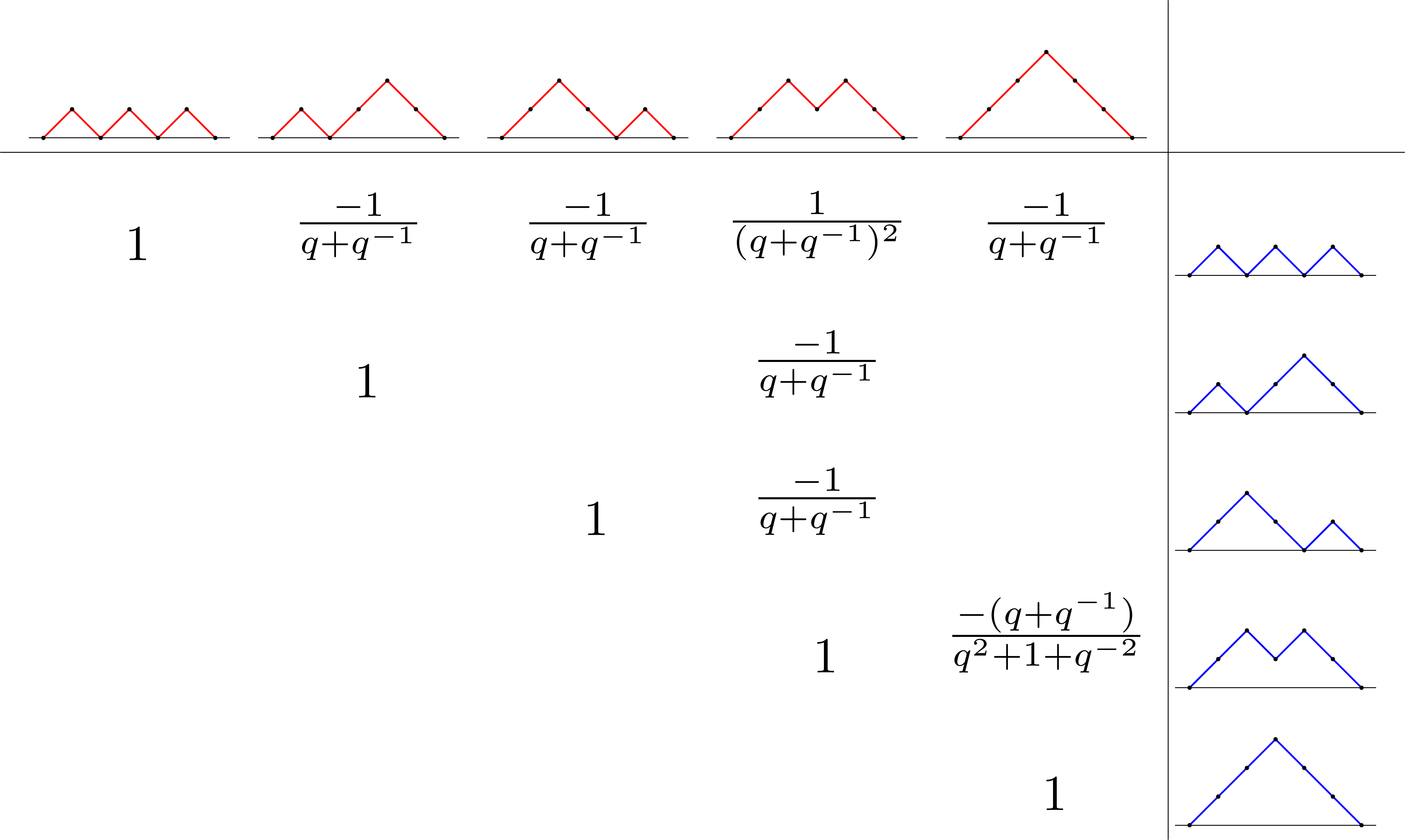} 
\caption{\label{fig: quantum KW matrices only}
The rows and columns of the matrix $\genMmat$ are indexed by Dyck
paths of $2N$ steps. The non-zero entries appear where a certain binary
relation --- the parenthesis reversal relation --- holds between the two Dyck paths.
This figure gives the explicit matrix elements
of $\genMmat$ in terms of $q = e^{\ii 4 \pi / \kappa}$
for 
$N=2$ and $N=3$.
}
\end{figure}
\begin{figure}
\centerfloat
\includegraphics[width = 0.336 \textwidth]{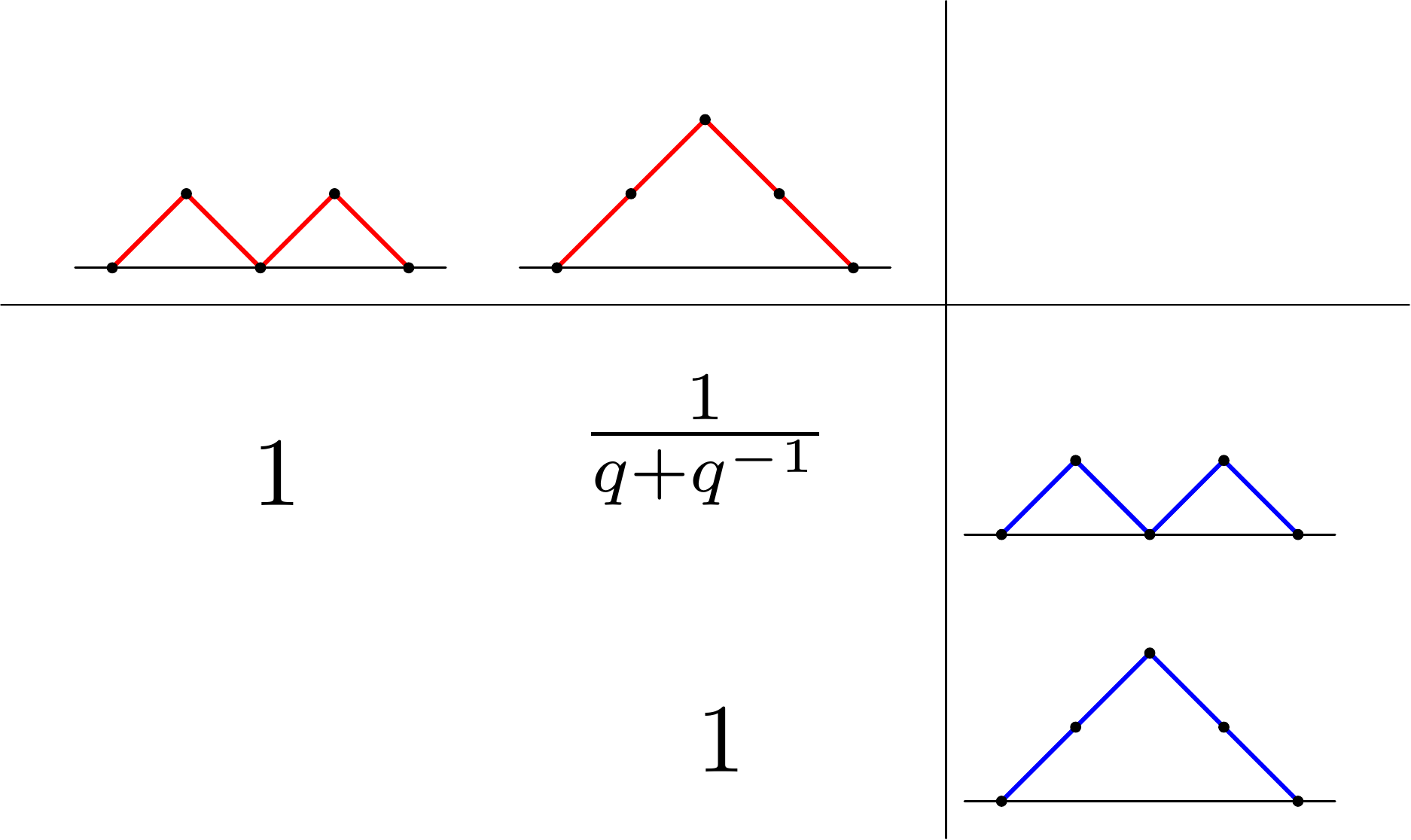} 
\hspace{0.5 cm}
\includegraphics[width = 0.54\textwidth]{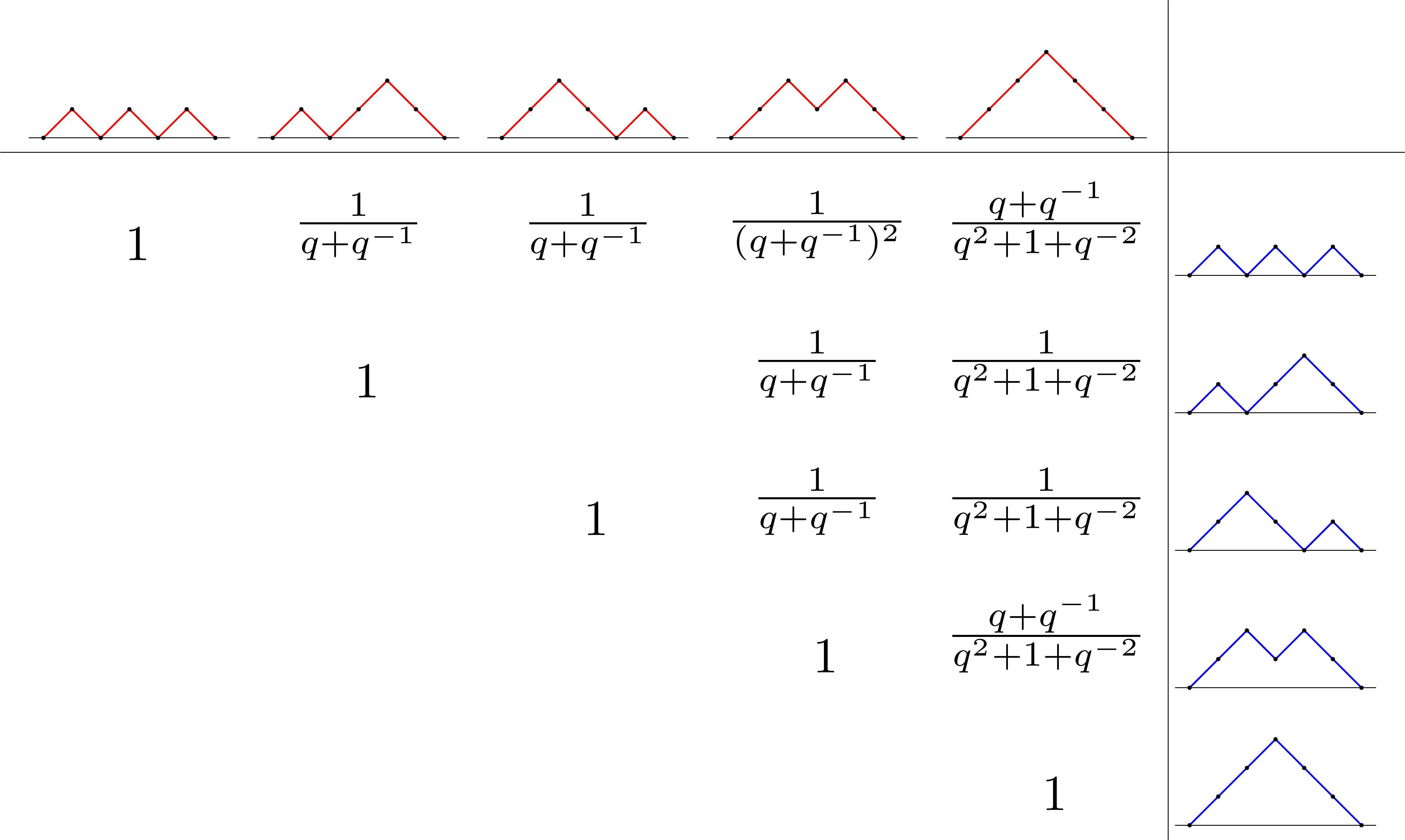} 
\caption{\label{fig: quantum CIDT matrices only}
The rows and columns of the matrix $\genMinv$ are indexed by Dyck
paths of $2N$ steps. The non-zero entries appear where the natural partial order
relation holds between the two Dyck paths: in particular, the matrix is upper triangular.
This figure gives the explicit matrix elements
of $\genMinv$ in terms of $q = e^{\ii 4 \pi / \kappa}$
for 
$N=2$ and $N=3$.
}
\end{figure}

The second main result of this article, Theorem~\ref{thm: conformal block functions}
given in Section~\ref{sec: construction of conformal block functions},
is a construction of the conformal block functions via the quantum group based method
of~\cite{KP-conformally_covariant_boundary_correlation_functions_with_a_quantum_group}.
Our construction expresses the conformal block functions as concrete linear combinations of
integrals of Coulomb gas type, similar to~\cite{DF-multipoint_correlation_functions}. 
It also reflects the underlying idea of conformal blocks, according to which
the Dyck path serves to label a sequence of intermediate representations.

\subsection*{Acknowledgments}
We thank Steven Flores and David Radnell for interesting discussions.

A.K. and K.K. are supported by the Academy of Finland project
``Algebraic structures and random geometry of stochastic lattice models''. A.K. is also supported by the Vilho, Yrj\"{o} and Kalle V\"{a}is\"{a}l\"{a} Foundation.
E.P. is supported by the ERC AG COMPASP, the NCCR SwissMAP, and the Swiss NSF.

\bigskip

\section{Combinatorial preliminaries}
\label{sec: combinatorial preliminaries}


In this section, we recall some combinatorial definitions and  results. A complete account can be found in 
our previous article~\cite[Section~2]{KKP-boundary_correlations_in_planar_LERW_and_UST}, whose notations and conventions we follow.

\subsection{Dyck paths, skew Young diagrams, and Dyck tiles}

We denote by $\DP_N$ the set of \emph{Dyck paths} of $2N$ steps, i.e., 
sequences $\alpha = (\alpha(0), \alpha(1), \ldots, \alpha(2N))$ such that $\alpha(j) \in \Znn$ and $|\alpha(j) - \alpha(j-1)| = 1$
for all 
$j \in \set{1,\ldots,2N}$, 
and $\alpha(0) = \alpha(2N)=0$. The number of such Dyck paths is a Catalan number,
\begin{align*}
\# \DP_N = \Catalan_N = \frac{1}{N+1} \binom{2N}{N} .
\end{align*}
We also denote by $\DP := \bigsqcup_{N \in \bZnn} \DP_N$ the set of
Dyck paths of arbitrary length.

For each $N$, the set of Dyck paths of $2N$ steps has a natural partial ordering:
for $\alpha, \beta \in \DP_N$ we denote $\alpha \DPleq \beta$
if and only if $\alpha(j) \leq \beta(j)$ for all 
$j \in \set{0,1,\ldots,2N}$.
When $\alpha \DPleq \beta$, the area between the Dyck paths $\alpha$ and $\beta$
forms a \emph{skew Young diagram}, denoted by $\alpha / \beta$.

The main combinatorial objects for the present article are certain tilings of
skew Young diagrams,
called Dyck tilings.
The tiles $t$ in these tilings are skew Young diagrams of a particular
type: namely $t = \alpha/\beta$ such that for some $0 < x_t \leq x'_t < 2N$
and $h_t \in \Zpos$ we have
\begin{align*}
& \begin{cases}
\alpha(j) = \beta(j) & \text{ for } 0 \leq j < x_t \\
\alpha(j) = \beta(j) -2 & \text{ for } x_t \leq j \leq x'_t \\
\alpha(j) = \beta(j) & \text{ for } x'_t < j \leq 2N 
\end{cases}
\end{align*}
and
\begin{align*}
& \alpha(x_t-1) = \beta(x_t-1) = \alpha(x_t+1) = \beta(x_t+1) = h_t .
\end{align*}
Such tiles $t=\alpha/\beta$ are called \emph{Dyck tiles}, the number $h_t$
is called the \emph{height} of $t$, and the intervals $[x_t , x_t']$
and $(x_t-1, x_t+1)$ are called the \emph{horizontal extent} and \emph{shadow} of $t$,
respectively. Figure~\ref{fig: extent shadow and height} illustrates these notions.
We say that a Dyck tile $t_2 = \alpha_2 / \beta_2$ \emph{covers} another Dyck
tile $t_1 = \alpha_1 / \beta_1$ if there exists a $j$ such that
$j \in [x_{t_1} , x'_{t_1}] \cap [x_{t_2} , x'_{t_2}]$ and $\alpha_1(j) < \alpha_2(j)$.
\begin{figure}[h!]
\includegraphics[width = 0.4\textwidth]{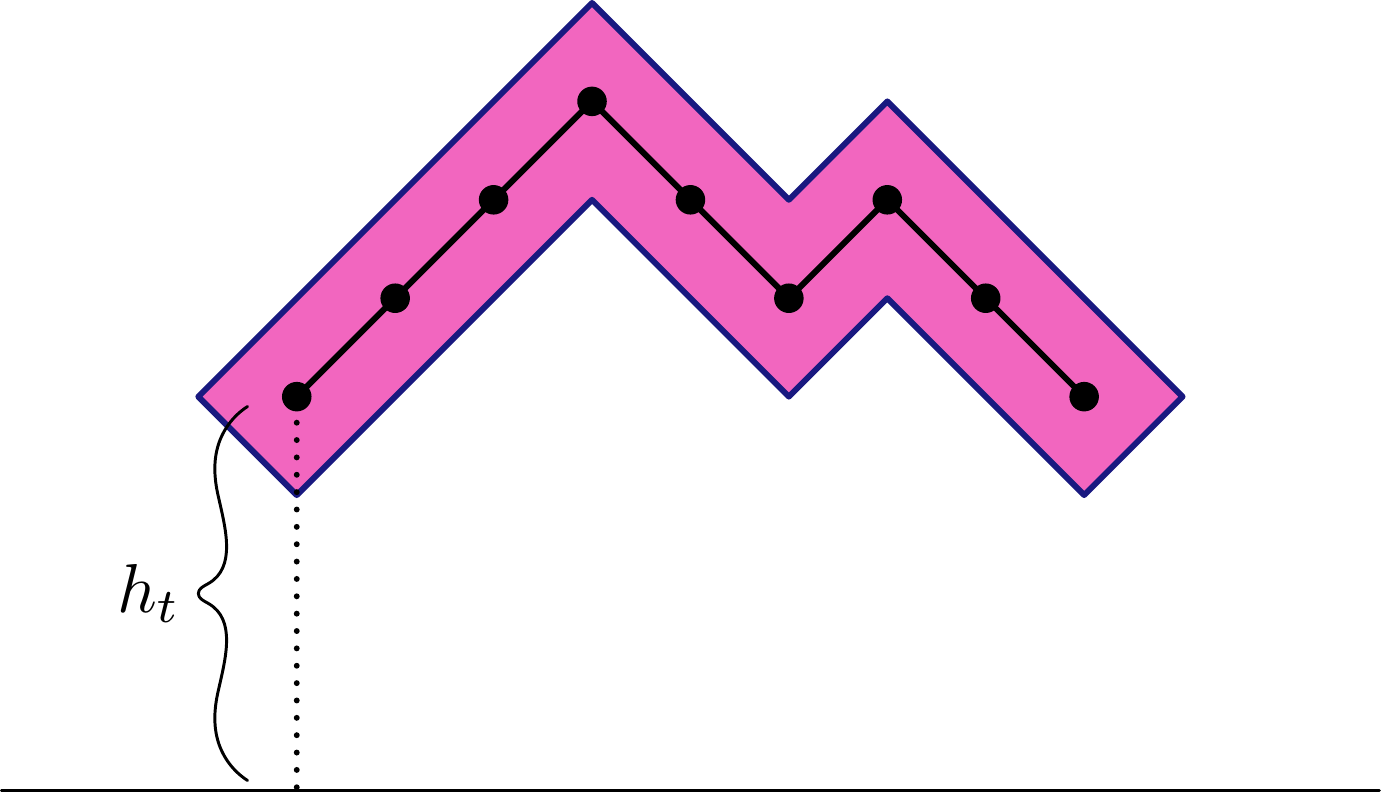} \hspace{1cm}
\includegraphics[width = 0.4\textwidth]{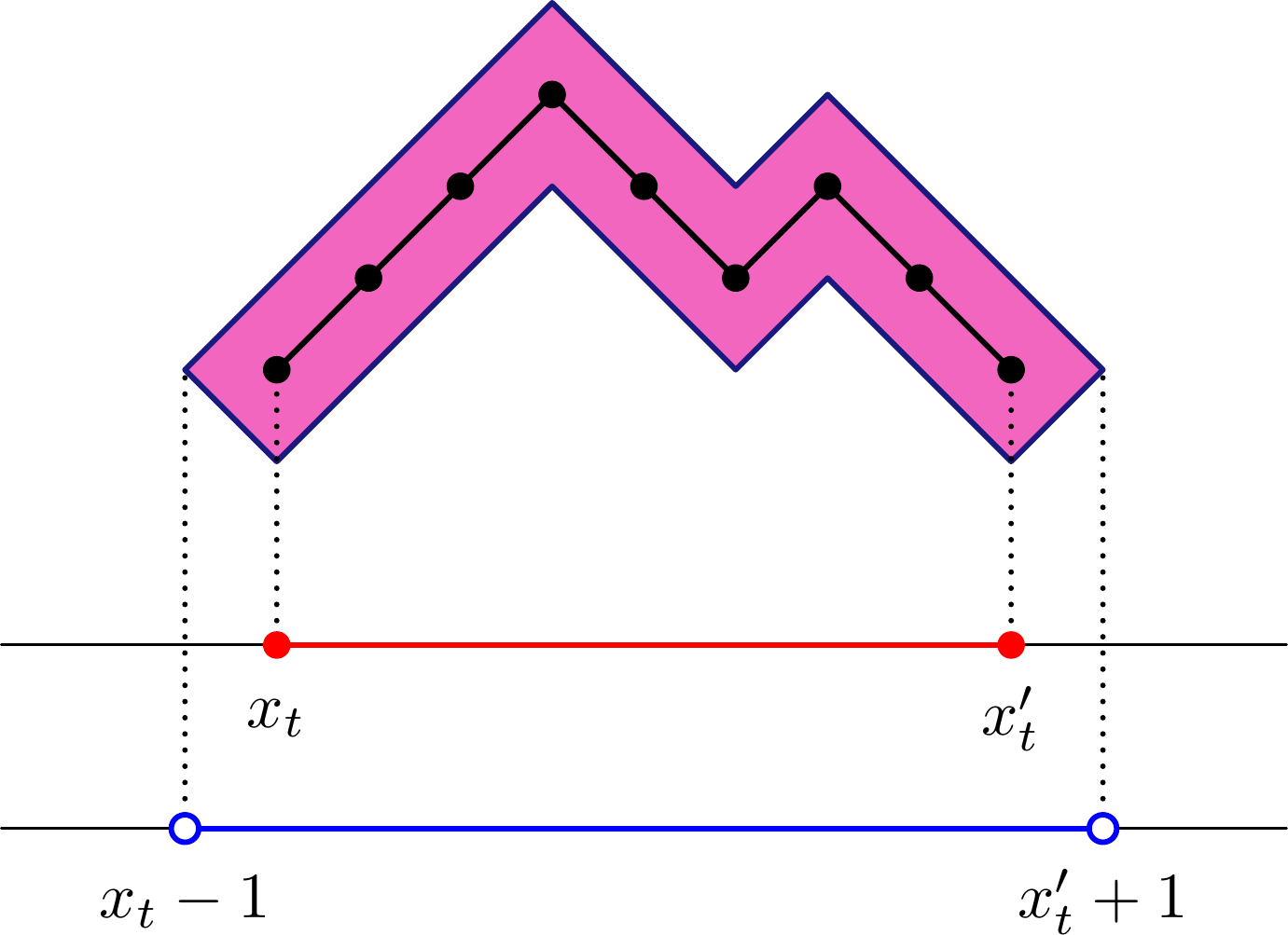}
\caption{\label{fig: extent shadow and height} The vertical position of a
Dyck tile $t$ is described by the integer height $h_t$.
The horizontal extent $[x_t , x'_t]$ (in red) and shadow $(x_t-1 , x'_t +1)$
(in blue) are intervals that describe the 
horizontal position.
The shape of a Dyck tile is essentially that of a 
Dyck path, as illustrated by the black path drawn inside the tile.
}
\end{figure}

In general, a \emph{Dyck tiling} $T$ of a skew Young diagram $\alpha/\beta$ is a collection
of Dyck tiles $t$ which cover the area of the skew Young diagram,
$\bigcup_{t \in T} t = \alpha/\beta$, and which have no overlap. 
Specifically, we consider
so called nested Dyck tilings and cover-inclusive Dyck tilings illustrated
in Figures~\ref{fig: NDT examples} and~\ref{fig: CIDT examples} 
and defined below 
in Sections~\ref{sub: nested  DT} and~\ref{sub: cover-inclusive DT}, respectively.

\subsection{Nested Dyck tilings and the parenthesis reversal relation}
\label{sub: nested  DT}

A Dyck tiling $T$ of a skew Young diagram $\alpha/\beta$ is said to be a \emph{nested
Dyck tiling} if the shadows of any two distinct tiles of $T$ are either disjoint
or one contained in the other, and in the latter case the tile with the larger
shadow covers the other. Figure~\ref{fig: NDT examples} exemplifies nested Dyck
tilings. It is not difficult to see that if a skew Young diagram $\alpha/\beta$ admits a
nested Dyck tiling, such a tiling is necessarily unique.
In this case, we write \[ \alpha \KWleq \beta , \] and 
we denote the nested Dyck tiling of $\alpha/\beta$ 
by $\nestedtilingof (\alpha / \beta)$.
This binary relation~$\KWleq$ on $\DP_N$ was first introduced
in~\cite{KW-double_dimer_pairings_and_skew_Young_diagrams, SZ-path_representations_of_maximal_paraboloc_KL_polynomials},
and we 
call it the \emph{parenthesis reversal relation}, because of a
convenient characterization it has in terms of balanced parenthesis expressions,
see~\cite[Lemma~2.7]{KKP-boundary_correlations_in_planar_LERW_and_UST}.
\begin{figure}
\includegraphics[width = 0.4\textwidth]{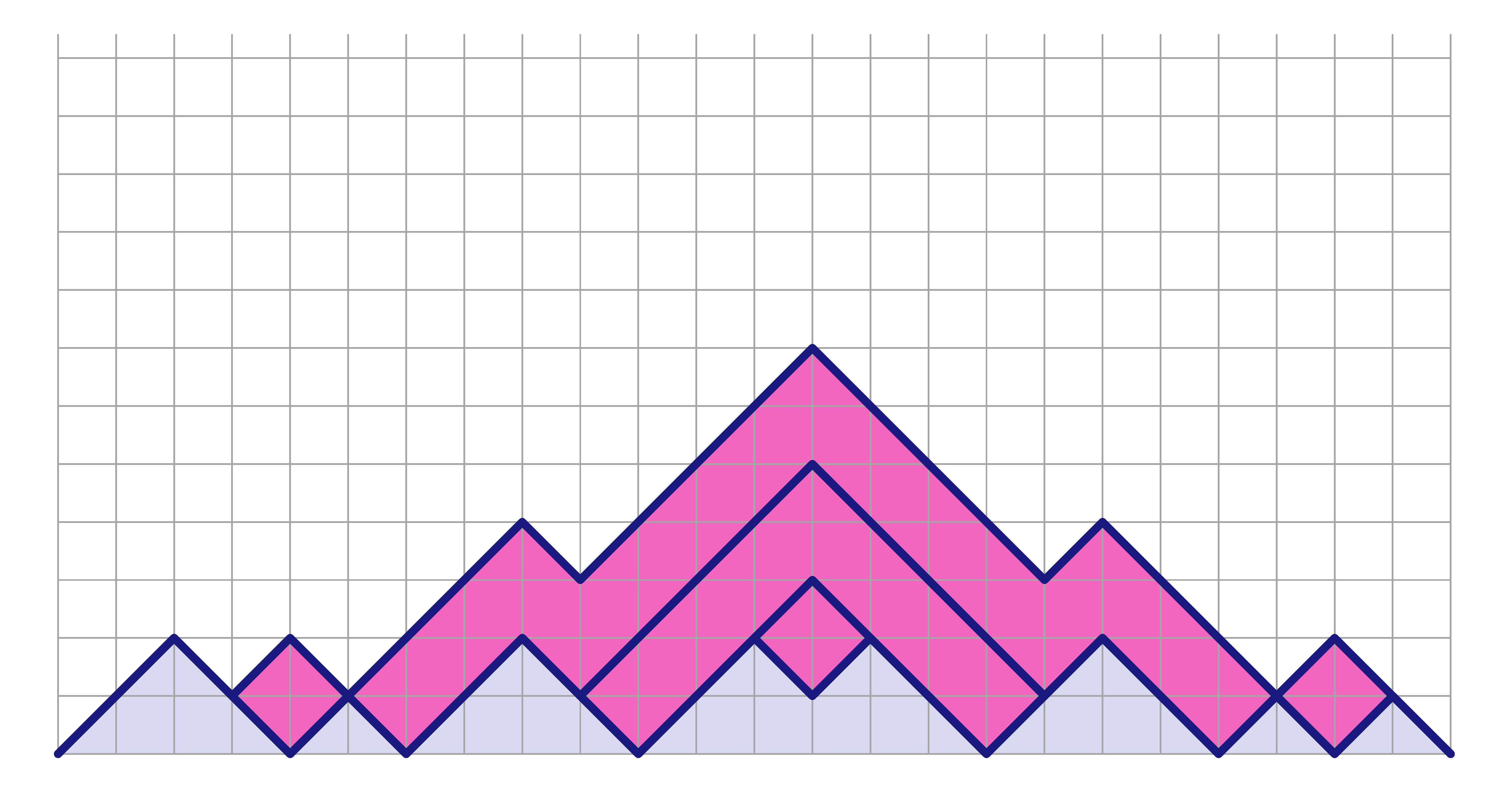} \quad
\includegraphics[width = 0.4\textwidth]{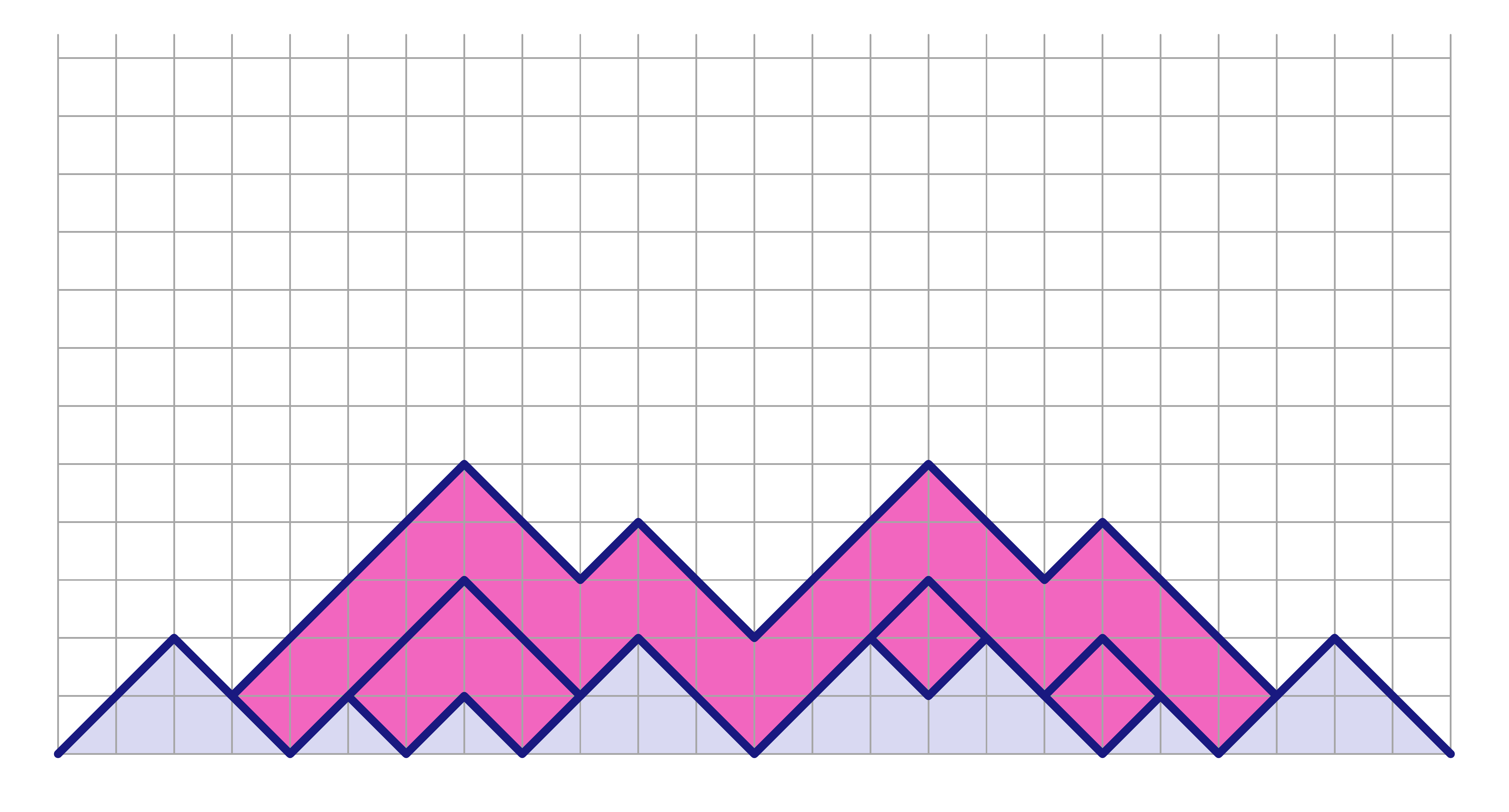}
\caption{\label{fig: NDT examples} Examples of nested Dyck tilings of skew Young diagrams.}
\end{figure}

\subsection{Cover-inclusive Dyck tilings}
\label{sub: cover-inclusive DT}

A Dyck tiling $T$ of a skew Young diagram $\alpha/\beta$ is said to be a \emph{cover-inclusive
Dyck tiling} if for any two distinct tiles of $T$, either the horizontal
extents are disjoint, or the tile that covers the other has horizontal extent contained in the horizontal
extent of the other. Figure~\ref{fig: CIDT examples} exemplifies cover-inclusive
Dyck tilings. In contrast with nested Dyck tilings, any skew Young diagram 
has cover-inclusive Dyck tilings.
For $\alpha \DPleq \beta$, the set of cover-inclusive Dyck tilings of $\alpha/\beta$ is
denoted by $\CItilingsof (\alpha/\beta)$.
\begin{figure}
\includegraphics[width = 0.4\textwidth]{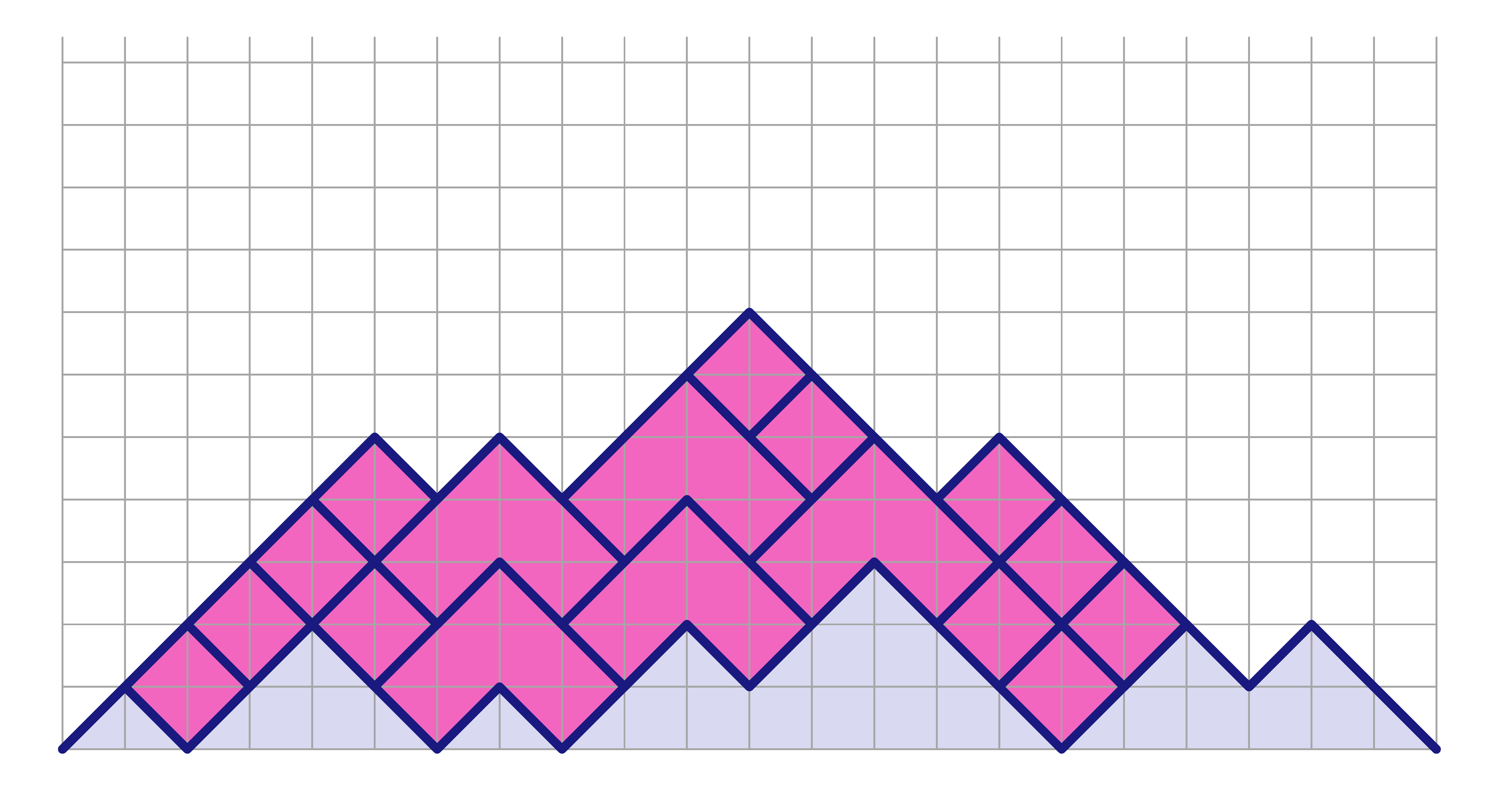} \quad
\includegraphics[width = 0.4\textwidth]{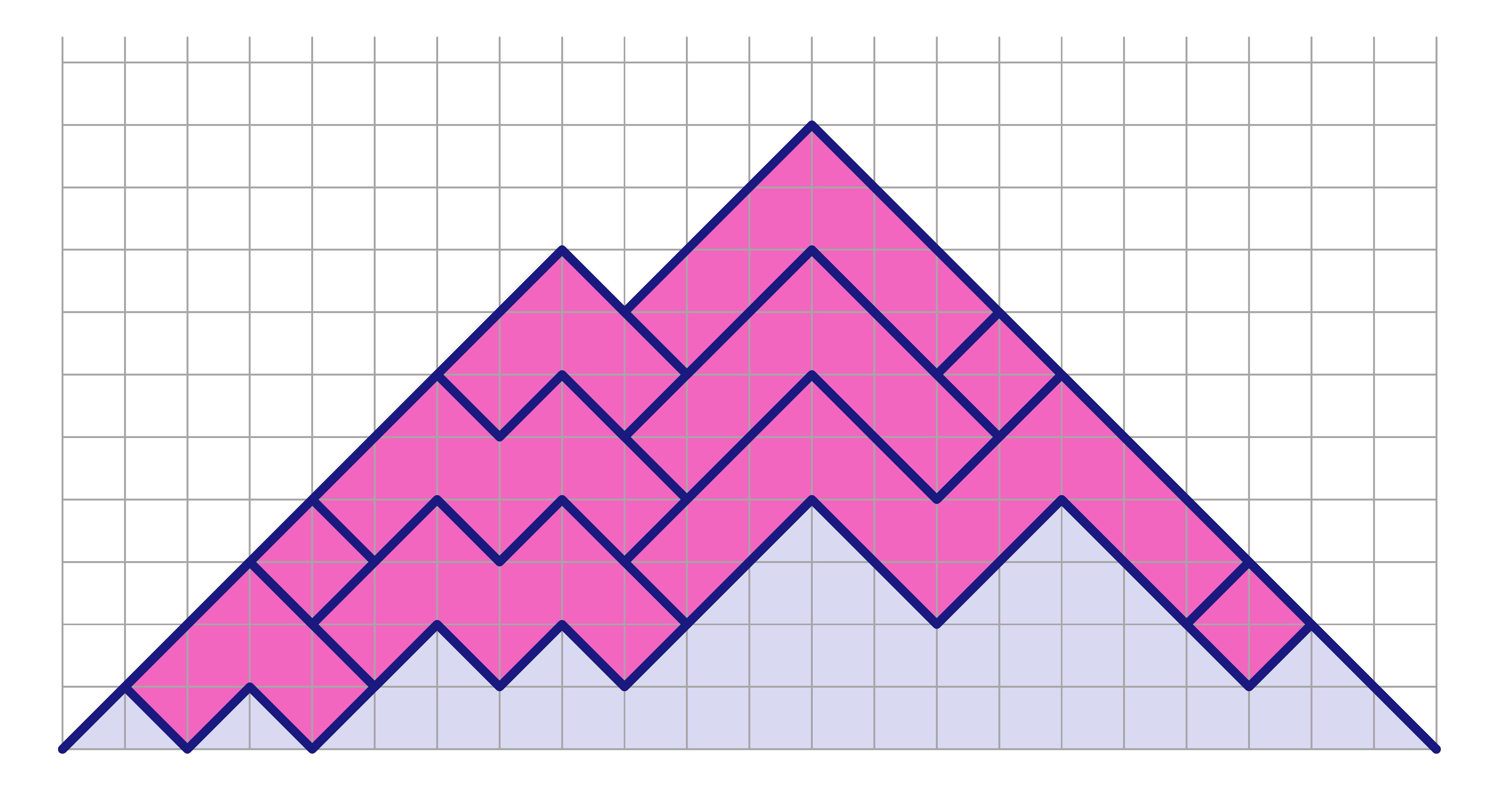} \\
\caption{\label{fig: CIDT examples} Examples of cover-inclusive Dyck tilings of skew Young diagrams.}
\end{figure}

\subsection{Weighted incidence matrices and their inversion}

The incidence matrix of the binary relation~$\KWleq$ on the set $\DP_N$ of Dyck paths
plays a role in the combinatorics of dimers and groves~\cite{KW-double_dimer_pairings_and_skew_Young_diagrams},
and of uniform spanning tree boundary branches~\cite{KKP-boundary_correlations_in_planar_LERW_and_UST}.
The rows and columns of this incidence matrix are indexed by Dyck paths, and its entries
are $1$ or $0$ 
according to whether 
or not
the relation $\KWleq$ holds between the two paths.
It turns out that an appropriately weighted incidence matrix is relevant for the combinatorics of conformal blocks. 

Suppose that a weight $w(t) \in \bC$ has been 
assigned
to each 
Dyck tile $t$.
We define the weighted incidence matrix by setting
for all $\alpha, \beta \in \DP_N$
\begin{align} \label{eq: def of weighted incidence matrix}
M_{\alpha, \beta} := \begin{cases}
\prod_{t \in \nestedtilingof (\alpha / \beta)} (-w(t)) & \text{if }\alpha \KWleq \beta\\
0 & \text{otherwise}, 
\end{cases} 
\end{align}
where $\nestedtilingof (\alpha / \beta)$ denotes the unique nested tiling of the skew Young diagram $\alpha / \beta$ when $\alpha \KWleq \beta$.


We rely on the following combinatorial result, which
gives a formula for the inverse of the weighted incidence 
matrix~\eqref{eq: def of weighted incidence matrix} in terms of cover-inclusive Dyck tilings.
Such formulas for the inverses appear
in~\cite{KW-double_dimer_pairings_and_skew_Young_diagrams,
SZ-path_representations_of_maximal_paraboloc_KL_polynomials, KKP-boundary_correlations_in_planar_LERW_and_UST}.
\begin{prop}
\label{prop: weighted KW incidence matrix inversion}
The weighted incidence matrix $M \in \bC^{\DP_N \times \DP_N}$ with entries~\eqref{eq: def of weighted incidence matrix}
is invertible, and the entries of the inverse matrix $M^{-1}$ are given by the weighted sums 
\begin{align*} 
M^{-1}_{\alpha, \beta} = \begin{cases}
\sum_{T \in \CItilingsof (\alpha/\beta)} \prod_{t \in T} w(t) & \text{if }\alpha \DPleq \beta \\
0 & \text{otherwise} 
\end{cases}
\end{align*}
over the sets 
$\CItilingsof (\alpha/\beta)$ 
of cover-inclusive Dyck tilings of the skew Young diagrams $\alpha/\beta$.
\end{prop}
\begin{proof}
In this form, the assertion is proved in~\cite[Theorem~2.9]{KKP-boundary_correlations_in_planar_LERW_and_UST}.
\end{proof}

\subsection{Slopes and wedges in Dyck paths and a recursion for incidence matrices}

Any two consecutive steps of a Dyck path $\alpha$ are said to form either a slope or a wedge,
according to the cases illustrated in Figure~\ref{fig: wedges and slopes}:
we say that $\alpha$ has a \emph{wedge} at $j$ if $\alpha(j-1) = \alpha(j+1)$,
and that $\alpha$ has a \emph{slope} at $j$ otherwise.

A slope at $j$ 
is called an \emph{up-slope} if $\alpha(j+1) = \alpha(j-1)+2$
and a \emph{down-slope} if $\alpha(j+1) = \alpha(j-1)-2$.
Without specifying the type of the slope, we denote
the presence of a slope at $j$ 
by $\slopeat{j} \in \alpha$.


A wedge at $j$ is called an \emph{up-wedge} if $\alpha(j) = \alpha(j \pm 1) + 1$, and 
a \emph{down-wedge} if $\alpha(j) = \alpha(j \pm 1) - 1$,
and in these two cases we respectively write $\upwedgeat{j} \in \alpha$ and $\downwedgeat{j} \in \alpha$.
Without specifying the type of the wedge, we denote
the presence of a wedge at $j$ 
by $\wedgeat{j} \in \alpha$.
By removing a wedge at $j$ from a Dyck path $\alpha \in \DP_N$ we obtain a shorter Dyck path
$\hat{\alpha} \in \DP_{N-1}$, namely
$\hat{\alpha} = (\alpha(0) , \alpha(1) , \ldots , \alpha(j-1) , \alpha(j+2) , \ldots , \alpha(2N))$.
According to whether the removed wedge is an up-wedge or a down-wedge, we 
write $\hat{\alpha} = \alpha \removeupwedge{j}$ or $\hat{\alpha} = \alpha \removedownwedge{j}$,
or without specifying the type of the removed wedge, we may
write $\hat{\alpha} = \alpha \removewedge{j}$.
\begin{figure}
\centering
\subfigure[up-wedge]
{\qquad \includegraphics[width = 0.12 \textwidth]{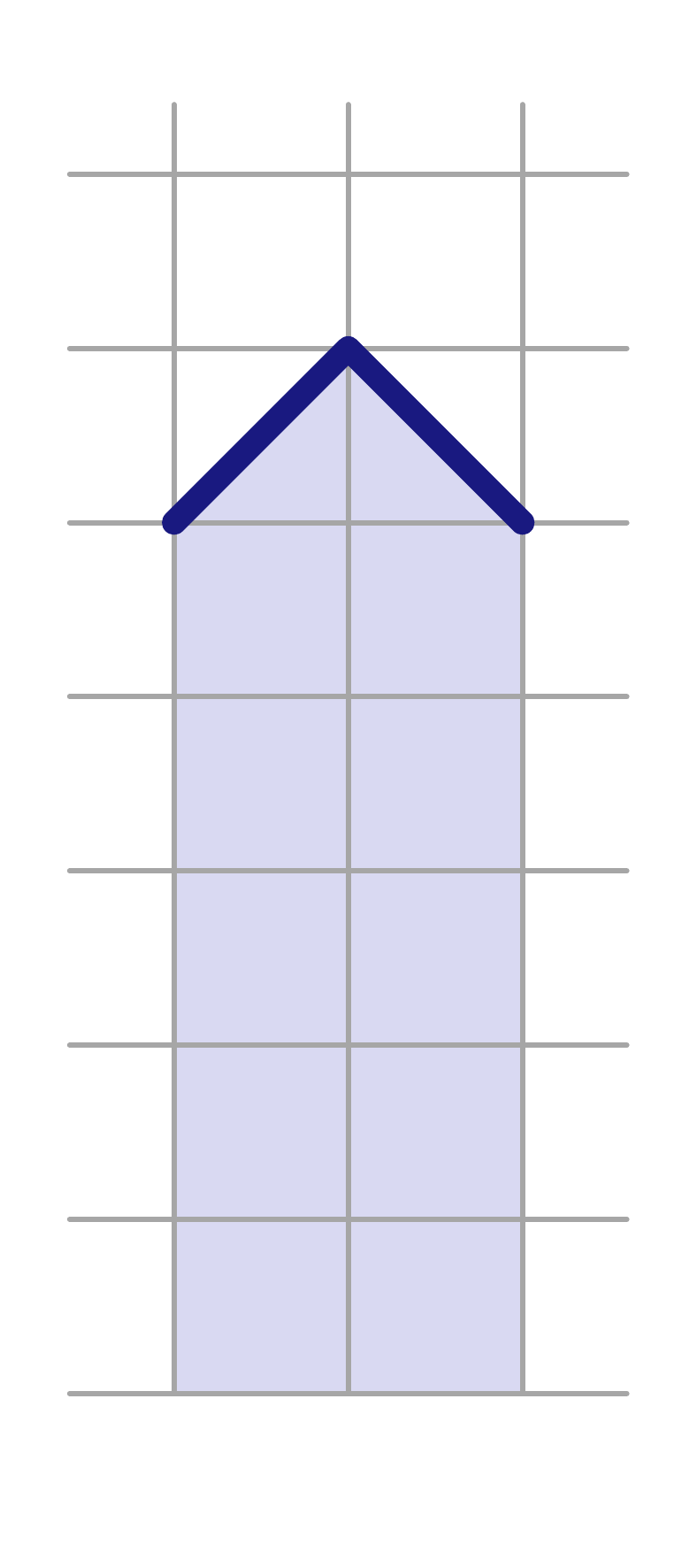} \qquad \label{fig: up-wedge}}
\subfigure[down-wedge]
{\qquad \includegraphics[width = 0.12 \textwidth]{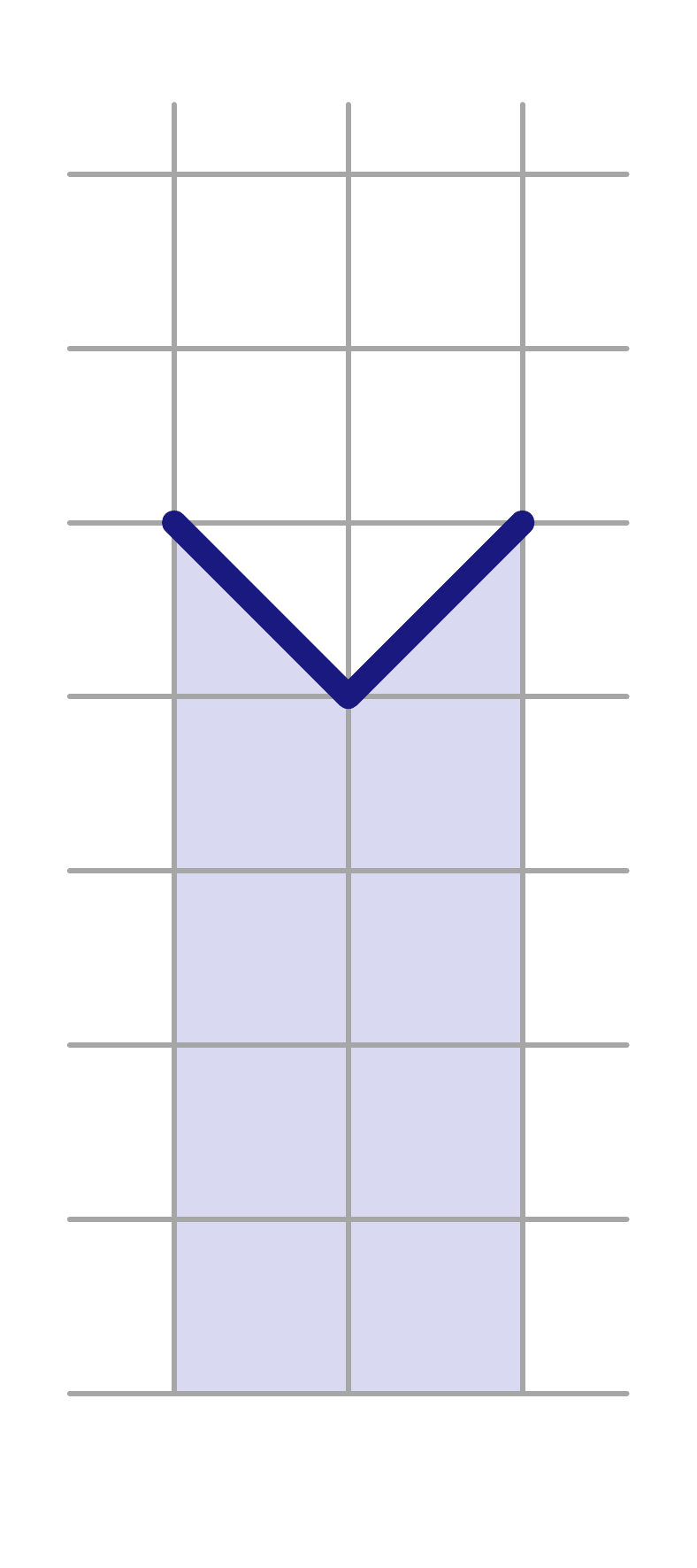} \qquad \label{fig: down-wedge}}
\subfigure[up-slope]
{\qquad \includegraphics[width = 0.12 \textwidth]{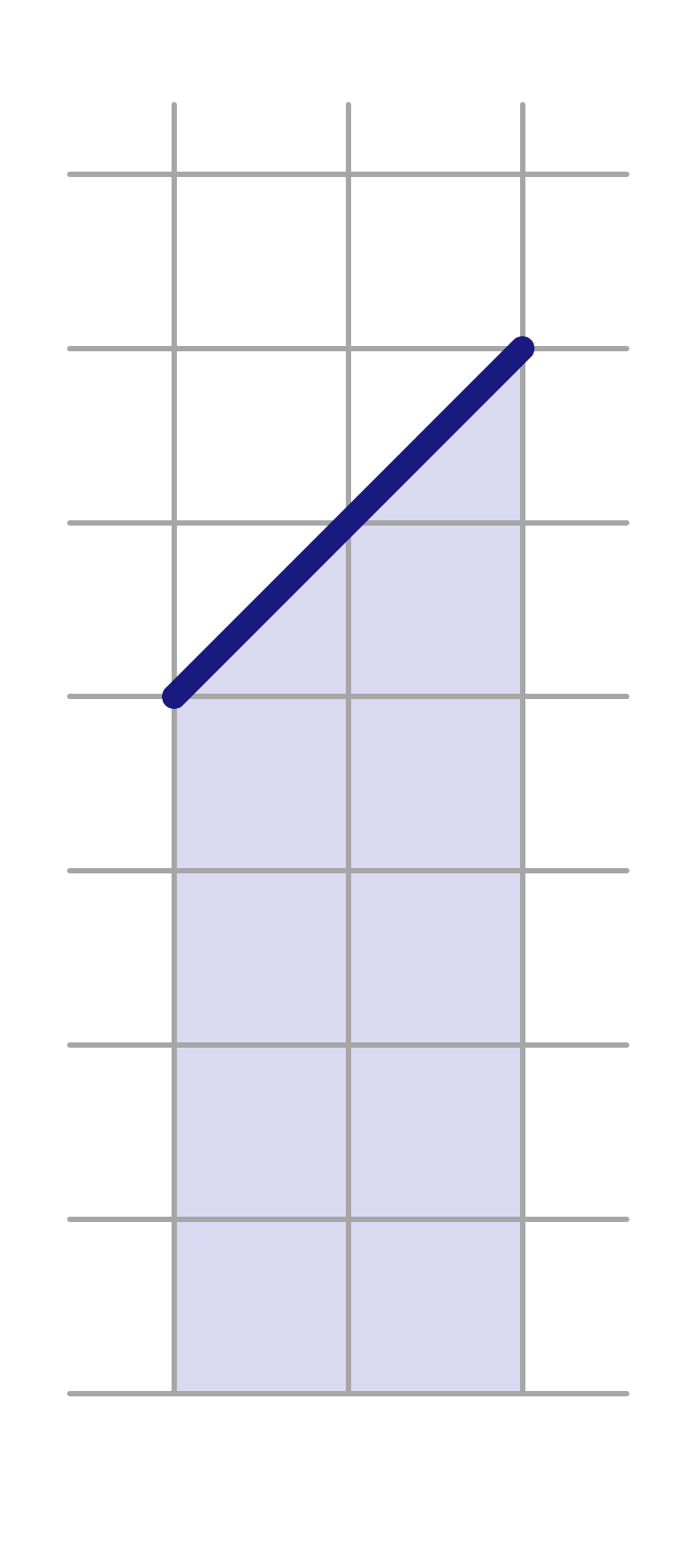} \qquad \label{fig: up-slope}}
\subfigure[down-slope]
{\qquad \includegraphics[width = 0.12 \textwidth]{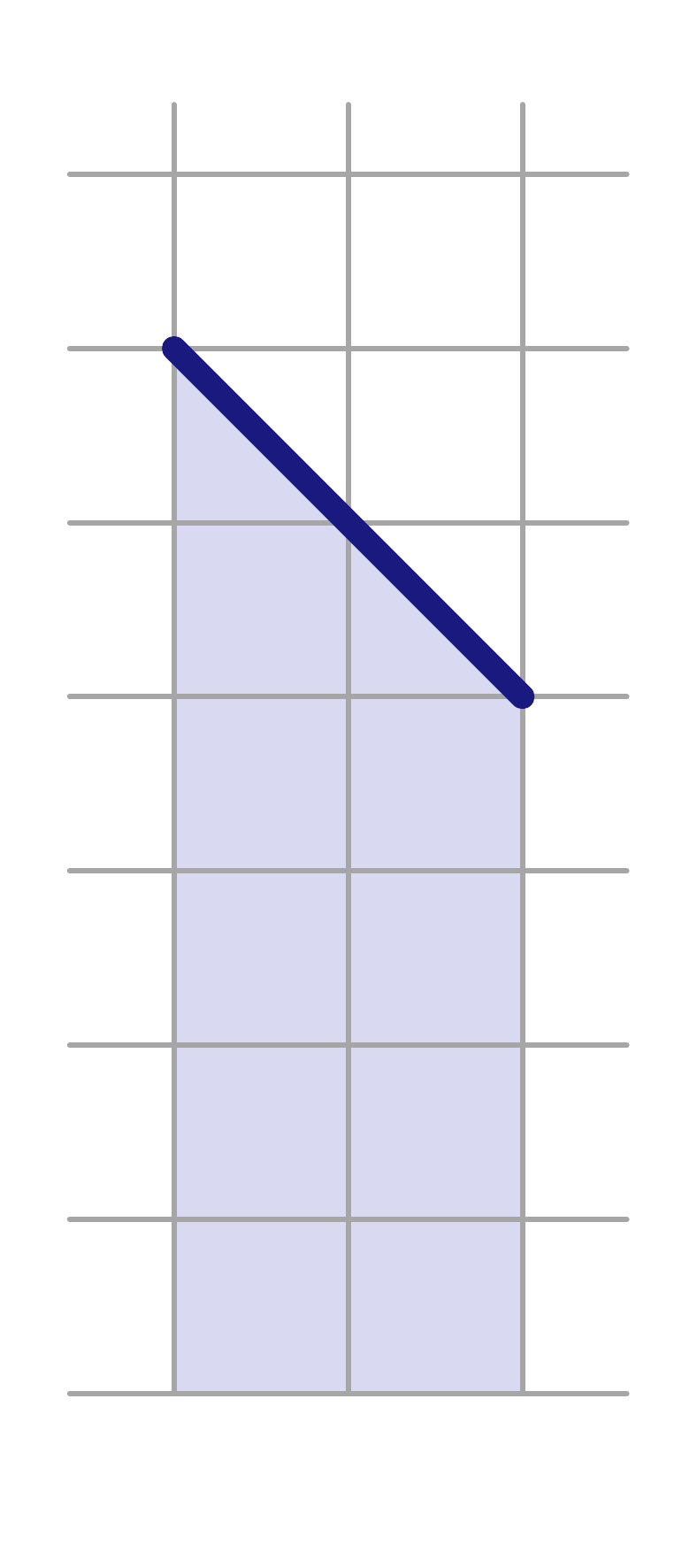} \qquad \label{fig: down-slope}}
\caption{\label{fig: wedges and slopes}
Wedges and slopes.
}
\end{figure}

Suppose that the weights $w(t)$ of Dyck tiles $t$ are chosen to only depend on
the height $h_t$ of the tile. 
Then, wedge removals 
allow for a characterization of weighted incidence matrices 
of the parenthesis reversal relation by the following recursion.
\begin{prop}\label{prop: recursion for matrix elements}
Let $f \colon \Zpos \to \bC$ be a given function. Then
the collection $(M^{(N)})_{N \in \bN}$ of weighted incidence
matrices~\eqref{eq: def of weighted incidence matrix} 
with weights of tiles determined by tile heights via $w(t) = f(h_t)$
is the unique collection of matrices $M^{(N)} \in \bC^{\DP_N \times \DP_N}$
satisfying the following recursion: we have $M^{(0)} = 1$,
and for any $N \in \bZpos$, and $\alpha, \beta \in \DP_N$, and $j \in \set{1,\ldots,2N-1}$ such that 
$\upwedgeat{j} \in \beta$, we have
\begin{align*} 
M_{\walk,\beta}^{(N)} = \begin{cases} 
0 & \text{if } \slopeat{j} \in \walk \\
M_{\hat{\walk},\hat{\beta}}^{(N-1)}
& \text{if } \upwedgeat{j} \in \walk \\
- f(\walk(j)+1) \times
M_{\hat{\walk},\hat{\beta}}^{(N-1)} 
& \text{if } \downwedgeat{j} \in \walk,
\end{cases} 
\end{align*}
where we denote by $\hat{\walk} = \walk \removewedge{j} \in \DP_{N-1}$ and 
$\hat{\beta} = \beta \removeupwedge{j} \in \DP_{N-1}$.
\end{prop}
\begin{proof}
See~\cite[Lemmas~2.13~and~2.14]{KKP-boundary_correlations_in_planar_LERW_and_UST}.
\end{proof}

\bigskip

\section{Conformal block functions}
\label{sec: conformal block functions background}


In the operator formalism of quantum field theories, fields correspond
to linear operators on the state space of the theory, and correlation
functions are written as ``vacuum expected values''. Somewhat more
concretely, $n$-point correlation functions are particular matrix
elements of a composition of $n$ linear operators on the state space.
Since the state space carries representations of the symmetries of the quantum field theory and
can be split into a direct sum of subrepresentations, 
it is natural to split these linear operators into corresponding blocks. 
In conformal field theory (CFT), 
the state space is a representation of the Virasoro algebra by virtue of conformal symmetry. 
The term conformal block 
refers to the idea of splitting the field operators into pieces
that go from one Virasoro subrepresentation of the state space to another, and compositions of field operators to
pieces that pass through a given sequence of 
subrepresentations, see~\cite{BPZ-infinite_conformal_symmetry_in_2D_QFT, Felder-BRST_approach,
DMS-CFT, Ribault-conformal_field_theory_on_the_plane}.

The main purpose of this section is to provide background for the definition of conformal block functions
that we use in the rest of the article. 
This definition is given in
Section~\ref{sub: defining conformal block functions}.
The background is included 
to provide sufficient context and main ideas, but its
presentation here is not intended to be fully rigorous.
We believe that a complete mathematical derivation of our defining properties 
is possible using the formalism of vertex operator
algebras~\cite{Lepowsky_Li-VOA}, but it is beyond the present work, and 
seems not to be readily contained in the existing literature.


\subsection{Highest weight representations of Virasoro algebra}


Usually, in
conformal field theory the state space is assumed to split into a
direct sum of highest weight representations of the Virasoro 
algebra~\cite{Kac-ICM_proceedings_Helsinki, FF-representations, IK-representation_theory_of_the_Virasoro_algebra},
with a common central charge $c\in\bR$ and various highest weights
$h\in\bR$. We parametrize the central charges $c\leq1$ 
by $\kappa>0$ via~\eqref{eq: central charge parametrization},
as is relevant to the theory of $\SLEk$ type random curves.
A special role is played by a primary field of conformal
weight $h=h_{1, 2} = \frac{6-\kappa}{2\kappa}$, 
which through fusion generates the so called first row of the Kac table with conformal weights 
\begin{align*} 
h(\lambda) = \; & h_{1,\lambda+1}=\frac{\lambda^{2}+2\lambda}{\kappa}-\frac{\lambda}{2} \qquad 
\text{for }
\lambda\in\Znn.
\end{align*}
In this article, we consider the generic case
$\kappa\notin\bQ$. 
Then, the irreducible highest weight representation with highest weight $h(\lambda)$ and
central charge $c$ is a quotient of the corresponding Verma module
by a submodule that itself is a Verma module\footnote{At rational values of $\kappa$, the structure of highest weight representations
can be more involved, see, e.g.,~\cite{IK-representation_theory_of_the_Virasoro_algebra}.}. 
We denote this irreducible
quotient 
by $\Qrep_{\lambda}$ and a highest weight vector in it by $\hwvec_{\lambda}$.
%
The contragredient (graded dual) representation $\Qrep_{\lambda}^{*}$
(see, e.g.,~\cite{IK-representation_theory_of_the_Virasoro_algebra}) is isomorphic
to $\Qrep_{\lambda}$, and we choose a highest weight vector $\hwvec_{\lambda}^{*}$
for it so that the normalization $\dualpairing{\hwvec_{\lambda}^{*}}{\hwvec_{\lambda}}=1$ holds.

\subsection{Conformal blocks}

\subsubsection{\textbf{Intertwining relation for primary field operators}}
A primary field $\psi$ of conformal weight $h$ is characterized by its
transformation property
\begin{align*}
\psi(x)\rightsquigarrow\;\confmap'(x)^{h}\;\psi(\confmap(x)) 
\end{align*}
under conformal transformations $\confmap$.
According to the seminal work~\cite{BPZ-infinite_conformal_symmetry_in_2D_QFT},
more general fields can be understood in terms of these primary fields.

In the operator formalism, a primary field $\psi(x)$ is realized by a primary field
operator $\Psi(x)$. We wish to split $\Psi(x)$
into conformal blocks between various highest weight representations $\Qrep_{\lambda}$ and
$\Qrep_{\mu}$, with $\lambda,\mu\in\Znn$.
The intertwining relation
\begin{align}
L_{n}\,\Primary(x)-\Primary(x)\,L_{n}=\; & x^{1+n}\pder x\Primary(x)+(1+n)x^{n}h\,\Primary(x)\label{eq: intertwining relation}
\end{align}
with the Virasoro generators $L_{n}$, $n \in \bZ$, is the infinitesimal form of
the primary field transformation property
under a conformal transformation $\confmap$ obtained
by varying the identity transformation to the direction of the holomorphic vector
field $\ell_{n}=-x^{1+n}\pder x$.

\subsubsection{\textbf{Matrix elements characterizing conformal blocks}}

The single matrix element between highest weight vectors,
\begin{align} \label{eq: 3pt function}
U_{\lambda}^{\mu}(x)=\; & \dualpairing{\hwvec_{\mu}^{*}}{\PrimaryBlock{\lambda}{\mu}(x)\,\hwvec_{\lambda}}, 
\end{align}
contains sufficient information to completely
determine $\PrimaryBlock{\lambda}{\mu}(x)$.
More generally, a composition of primary field operators
splits into conformal blocks indexed by a sequence $\sigmaseq=(\sigma_{0},\sigma_{1},\ldots,\sigma_{n})$
with $\sigma_{j}\in\Znn$ 
labeling the intermediate
representations $\Qrep_{\sigma_{0}},\Qrep_{\sigma_{1}},\ldots,\Qrep_{\sigma_{n}}$.
Now, the matrix element 
\begin{align}
U_{\sigmaseq}(x_{1},x_{2},\ldots,x_{n})=\; & \dualpairing{\hwvec_{\sigma_{n}}^{*}}{\PrimaryBlock{\sigma_{n-1}}{\sigma_{n}}(x_{n})\cdots\PrimaryBlock{\sigma_{1}}{\sigma_{2}}(x_{2})\,\PrimaryBlock{\sigma_{0}}{\sigma_{1}}(x_{1})\,\hwvec_{\sigma_{0}}}\label{eq: general conformal block function}
\end{align}
contains sufficient information to uniquely
determine the block of the composition.
Note that $U_{\lambda}^{\mu}(x)$ is a special case of $U_{\sigmaseq}(x_{1},x_{2},\ldots,x_{n})$
with $n=1$, $\sigma_{0}=\lambda$, and $\sigma_{1}=\mu$. 
We furthermore point out that $U_{\sigmaseq}(x_{1},x_{2},\ldots,x_{n})$ appears
in an actual vacuum expected value of $n$ fields if the highest weight
states on the right and left are the absolute vacua, i.e., if $\sigma_{0}=0$ and $\sigma_{n}=0$.

In the vertex operator algebra axiomatization of conformal field theory~\cite{Lepowsky_Li-VOA},
the block
$\PrimaryBlock{\lambda}{\mu}(x)$ is a formal power
series in $x$ with
coefficients that are linear operators, and 
the matrix elements
$U_{\lambda}^{\mu}(x)$ and $U_{\sigmaseq}(x_{1},x_{2},\ldots,x_{n})$ are formal power series with complex coefficients. For radially ordered variables
\begin{align*}
 & 0<|x_{1}|<|x_{2}|<\cdots<|x_{n}|,
\end{align*}
the series are in fact convergent, so we may
view~\eqref{eq: 3pt function} and~\eqref{eq: general conformal block function}
as actual functions.
The fact that they determine the operators and their compositions
justifies calling them conformal block functions.

\subsection{Properties of conformal block functions}

We now review properties of the conformal block functions $U_{\sigmaseq}(x_{1}, \ldots, x_{n})$,
which in particular completely fix the form of the matrix elements $U_{\lambda}^{\mu}(x)$,
and characterize for which $\lambda$ and $\mu$ the block $\PrimaryBlock{\lambda}{\mu}(x)$
can be non-vanishing in the first place.

\subsubsection{\textbf{Covariance properties}}

The intertwining relation~\eqref{eq: intertwining relation} for $L_0$
combined with the eigenvalues $L_{0} \, \hwvec_{\sigma_0} = h(\sigma_0)\,\hwvec_{\sigma_0}$
and $L_{0}^\top \, \hwvec^*_{\sigma_n} = h(\sigma_n)\,\hwvec_{\sigma_n}^*$
of the highest weight vectors gives
\begin{align*}
\big( h(\sigma_n) - h(\sigma_0) \big) \; U_{\sigmaseq}(x_{1},\ldots,x_{n})
    = \sum_{j=1}^n \Big( x_j \pder{x_j}  + h \Big) \; U_{\sigmaseq}(x_{1},\ldots,x_{n}) ,
\end{align*}
with $h=h(1)=\frac{6-\kappa}{2\kappa}$.
This infinitesimal relation can be integrated to 
obtain the homogeneity property
\begin{align} \label{eq: homogeneity of conformal block functions}
U_{\sigmaseq}(rx_{1},\ldots,rx_{n})=\; & r^{h(\sigma_{n})-h(\sigma_{0})-nh}\;U_{\sigmaseq}(x_{1},\ldots,x_{n})
    , \qquad\text{for } r>0 , 
\end{align}
of the conformal block functions. For the simplest conformal block
function $U_{\lambda}^{\mu}(x)$, this homogeneity fixes its
form
up to a multiplicative constant $C_{\lambda}^{\mu}$: 
\begin{align} \label{eq: the general form of 3pt function}
U_{\lambda}^{\mu}(x) =  \dualpairing{\hwvec_{\mu}^{*}}{\PrimaryBlock{\lambda}{\mu}(x)\,\hwvec_{\lambda}} =\; & C_{\lambda}^{\mu}\;x^{h(\mu)-h(\lambda)-h} .
\end{align}

Next, if we have $\sigma_0 = 0$, then
$L_{-1} \, \hwvec_{\sigma_0} = 0$ 
is a null vector in the quotient representation $\Qrep_0$.
Together with the intertwining relation~\eqref{eq: intertwining relation} for $L_{-1}$, and the property $L_{-1}^\top \, \hwvec^*_{\sigma_n} = 0$,
this gives the infinitesimal form of 
the translation invariance
\begin{align*} 
U_{\sigmaseq}(x_{1} + t, \ldots, x_{n} + t)
    = U_{\sigmaseq}(x_{1}, \ldots, x_{n})
    , \qquad\text{for } t \in \bR .
\end{align*}

Likewise, if $\sigma_n = 0$, then $L_1^\top \, \hwvec_{\sigma_n}^* = 0$
is a null vector in the contragredient representation $\Qrep_0^*$. Together with the intertwining relation~\eqref{eq: intertwining relation} for $L_1$, and the property $L_{1} \, \hwvec_{\sigma_0} = 0$,
this gives the infinitesimal form of the following covariance under special conformal transformations:
\begin{align*}
U_{\sigmaseq} \Big( \frac{x_{1}}{1 - s x_1}, \ldots, \frac{x_{n}}{1 - s x_n} \Big)
    = \prod_{j=1}^n (1 - s x_j)^{2h} \; \times  U_{\sigmaseq}(x_{1}, \ldots, x_{n})
    , \qquad\text{for } s \in \bR .
\end{align*}

For the conformal blocks that contribute to the vacuum expected value, we
have 
$\sigma_0 = 0$ and $\sigma_n = 0$. These conformal block functions satisfy the covariance
\begin{align*} 
U_{\sigmaseq}(x_{1}, \ldots, x_{n})
= \prod_{j=1}^n \Mob'(x_j)^{2h} \; \times U_{\sigmaseq} \big( \Mob(x_1), \ldots, \Mob(x_n) \big)
\end{align*}
under general M\"obius transformations $\Mob(x) = \frac{a x + b}{c x + d}$
with $a,b,c,d \in \bR$ and $ad-bc>0$.

\subsubsection{\textbf{Partial differential equations}}


Suppose now that the primary field $\Psi(x)$ of conformal weight $h=h(1)=\frac{6-\kappa}{2 \kappa}$ has the same 
degeneracy at grade 
two as the quotient representation $\Qrep_{1}$ of highest weight $h(1)$. 
Then the conformal block functions satisfy partial differential equations of second order.
These PDEs obtain a more symmetric expression in terms of the shifted
versions of the conformal block functions defined by
\begin{align*}
\widetilde{U}_{\sigmaseq}(x_{0},x_{1},\ldots,x_{n}) :=\; & U_{\sigmaseq}(x_{1}-x_{0},\ldots,x_{n}-x_{0}).
\end{align*}

The 
partial differential equation
arising from the degeneracy of $\Psi(x_j)$ at grade two takes the form
given in~\cite{BPZ-infinite_conformal_symmetry_in_2D_QFT},
\begin{align}\label{eq: PDEs for conformal block functions}
\begin{split}
\Bigg\{ & \; \frac{\kappa}{2}\pdder{x_{j}} + \frac{2}{x_{0}-x_{j}}\pder{x_{0}} - \frac{2\,h(\sigma_{0})}{(x_{0}-x_{j})^{2}} \quad 
	  \\
& \; \qquad + \sum_{\substack{i=1,\ldots,n \\ i\neq j}} \frac{2}{x_{i}-x_{j}}\pder{x_{i}}
- \sum_{\substack{i=1,\ldots,n\\i\neq j}} \frac{2\,h}{(x_{i}-x_{j})^{2}} \Bigg\} \; \widetilde{U}_{\sigmaseq}(x_{0},x_{1},\ldots,x_{n})
       = 0 . 
\end{split}
\end{align}

\subsubsection{\textbf{Selection rules}}

The PDEs above in particular imply selection rules for when non-vanishing conformal
blocks can exist. Namely, for $U_{\lambda}^{\mu}(x)=C_{\lambda}^{\mu}\,x^{\Delta}$,
with $\Delta=h(\mu)-h(\lambda)-h$ as in Equation~\eqref{eq: the general form of 3pt function},
the requirement of the PDE~\eqref{eq: PDEs for conformal block functions} amounts to
\begin{align*}
C_{\lambda}^{\mu}\;\left(\frac{\kappa}{2}\Delta(\Delta-1)+2\Delta -2h(\lambda)\right)\;x^{\Delta-2}=\; & 0,
\end{align*}
which implies either the vanishing of $C_{\lambda}^{\mu}$ and therefore
of the entire conformal block, or a quadratic equation relating the
conformal weights $h(\mu)$ and $h(\lambda)$. For fixed $\lambda$,
the two solutions of this quadratic equation are obtained at $\mu=\lambda\pm1$.
One can therefore conclude that the conformal blocks take the form
\begin{align}
U_{\lambda}^{\mu}(x)=\; & \begin{cases}
C_{\lambda}^{\pm}\;x^{h(\mu)-h(\lambda)-h} & \text{ if }\mu=\lambda\pm1\\
0 & \text{ if }|\mu-\lambda|\neq1.
\end{cases}\label{eq: form of 3pt function with selection}
\end{align}
The normalizations $C_{\lambda}^{\pm}$ are not canonically fixed;
in fact, the space of intertwining operators forms a vector space
(in the present cases always of dimension 
one or zero depending on whether 
the selection rules are fulfilled). 
One can make any convenient choice, and we will fix our choice later.

In the case $\lambda=0$, there is one further selection rule:
by translation invariance, 
$U_{0}^{\mu}(x)$ is constant. This further restricts the possibilities
in~\eqref{eq: the general form of 3pt function} to $h(\mu)=h=h(1)$, i.e., $\mu=1$.

In conclusion, a non-vanishing intertwining operator from $\Qrep_{\lambda}$
to $\Qrep_{\mu}$ can only exist if $|\mu-\lambda|=1$ and $\mu\geq0$.
Consequently, the composition of intertwining operators as in the
conformal block function~\eqref{eq: general conformal block function}
can only be non-trivial if the sequence $\sigmaseq=(\sigma_{0},\sigma_{1},\ldots,\sigma_{n})$
satisfies $\sigma_{j}\in\Znn$ and $|\sigma_{j}-\sigma_{j-1}|=1$
for all $j$. This means that non-trivial conformal block functions
are indexed by nearest neighbor walks on non-negative integers. 
The conformal block functions that contribute to the actual vacuum expected
value of $n$ fields must furthermore have $\sigma_{0}=0$ and $\sigma_{n}=0$,
so they are in fact indexed by
Dyck paths, and in particular, $n$ must be even.

\subsubsection{\textbf{Asymptotics}}

Above, in Equation~\eqref{eq: form of 3pt function with selection}, we completely
fixed the form of the simplest conformal block function $U_{\lambda}^{\mu}(x)$,
i.e., the case $n=1$. Now we consider the next case $n=2$ with $\sigmaseq=(\sigma_{0},\sigma_{1},\sigma_{2})$,
and the corresponding conformal block function
\begin{align*}
U_{\sigmaseq}(x_{1},x_{2})=\; & \dualpairing{\hwvec_{\sigma_{2}}^{*}}{\PrimaryBlock{\sigma_{1}}{\sigma_{2}}(x_{2})\,\PrimaryBlock{\sigma_{0}}{\sigma_{1}}(x_{1})\,\hwvec_{\sigma_{0}}}.
\end{align*}

If $x_{1}$ is kept fixed, then the vector $\PrimaryBlock{\sigma_{0}}{\sigma_{1}}(x_{1})\,\hwvec_{\sigma_{0}}$
can be expanded in the usual basis of $\Qrep_{\sigma_{1}}$ as
\begin{align*}
\PrimaryBlock{\sigma_{0}}{\sigma_{1}}(x_{1})\,\hwvec_{\sigma_{0}}
    = & \sum_{k \in \bN} \; \sum_{n_{1},\ldots,n_{k}>0} a_{n_{1},\ldots,n_{k}}(x_{1}) \; L_{-n_{k}}\cdots L_{-n_{1}}\,\hwvec_{\sigma_{1}},
\end{align*}
where in particular the coefficient of the highest weight vector $\hwvec_{\sigma_{1}}$
is picked by the projection to $\hwvec_{\sigma_{1}}^{*}$,
\begin{align*}
a_{\emptyset}(x_{1})=\; & \dualpairing{\hwvec_{\sigma_{1}}^{*}}{\PrimaryBlock{\sigma_{0}}{\sigma_{1}}(x_{1})\,\hwvec_{\sigma_{0}}}
    = U_{\sigma_{0}}^{\sigma_{1}}(x_{1}) = C_{\sigma_{0}}^{\sigma_{1}}\,x_{1}^{h(\sigma_{1})-h(\sigma_{0})-h}.
\end{align*}
Using this expansion, the conformal block function becomes
\begin{align*}
U_{\sigmaseq}(x_{1},x_{2})
    =\; & \sum_{k \in \bN} \; \sum_{n_{1},\ldots,n_{k}>0} a_{n_{1},\ldots,n_{k}}(x_{1})\;\dualpairing{\hwvec_{\sigma_{2}}^{*}}{\PrimaryBlock{\sigma_{1}}{\sigma_{2}}(x_{2})\,L_{-n_{k}}\cdots L_{-n_{1}}\,\hwvec_{\sigma_{1}}}.
\end{align*}
With the generic form~\eqref{eq: the general form of 3pt function} of the conformal blocks $U_{\sigma_{1}}^{\sigma_{2}}(x_{2})$,
the intertwining relation~\eqref{eq: intertwining relation} implies
\begin{align*}
\dualpairing{\hwvec_{\sigma_{2}}^{*}}{\PrimaryBlock{\sigma_{1}}{\sigma_{2}}(x_{2})\,L_{-n_{k}}\cdots L_{-n_{1}}\hwvec_{\sigma_{1}}} 
\propto x_{2}^{h(\sigma_{2}) - h(\sigma_{1}) - h - n_{1} - \, \cdots \, - n_k} .
\end{align*}
The leading contribution
in the limit $x_{2}\to\infty$ 
comes from the highest weight vector $\hwvec_{\sigma_1}$, since $n_j>0$.
Thus, for fixed $x_{1}$ and as $x_{2}\to\infty$, 
the leading asymptotics of the conformal block function is
\begin{align}
\begin{split}
U_{\sigmaseq}(x_{1},x_{2})
    \sim\; & a_{\emptyset}(x_{1})\;\dualpairing{\hwvec_{\sigma_{2}}^{*}}{\PrimaryBlock{\sigma_{1}}{\sigma_{2}}(x_{2})\,\hwvec_{\sigma_{1}}} \\
    =\; & a_{\emptyset}(x_{1})\times U_{\sigma_{1}}^{\sigma_{2}}(x_{2}) \\
    =\; & C_{\sigma_{0}}^{\sigma_{1}}\,x_{1}^{h(\sigma_{1})-h(\sigma_{0})-h}\times C_{\sigma_{1}}^{\sigma_{2}}\,x_{2}^{h(\sigma_{2})-h(\sigma_{1})-h}. \\
    =\; & C_{\sigma_{0}}^{\sigma_{1}}C_{\sigma_{1}}^{\sigma_{2}}\;x_{1}^{h(\sigma_{1})-h(\sigma_{0})-h}\,x_{2}^{h(\sigma_{2})-h(\sigma_{1})-h}.
    \label{eq: asymptotics for x2}
\end{split}
\end{align}
By a similar argument, for fixed $x_{2}$ and as $x_{1}\to0$,
the leading asymptotics of the conformal block function is
\begin{align} \label{eq: asymptotics for x1}
U_{\sigmaseq}(x_{1},x_{2})\sim\; & C_{\sigma_{0}}^{\sigma_{1}}C_{\sigma_{1}}^{\sigma_{2}}\;x_{2}^{h(\sigma_{2})-h(\sigma_{1})-h}\,x_{1}^{h(\sigma_{1})-h(\sigma_{0})-h}.
\end{align}
The remaining interesting asymptotics of $U_{\sigmaseq}(x_{1},x_{2})$
concerns the limit $|x_{1}-x_{2}|\to0$.
To analyze these, we resort to direct solutions of the PDEs~\eqref{eq: PDEs for conformal block functions} in the following.

The homogeneity~\eqref{eq: homogeneity of conformal block functions}
can be used to cast the $n=2$ conformal block function into the form
\begin{align*}
U_{\sigmaseq}(x_{1},x_{2})=\; & x_{2}^{h(\sigma_{2})-h(\sigma_{0})-2h}\;g_{\sigmaseq}\left(\frac{x_{1}}{x_{2}}\right),
\end{align*}
and the PDEs~\eqref{eq: PDEs for conformal block functions} for $\widetilde{U}_{\sigmaseq}(x_{0},x_{1},x_{2})=(x_{2}-x_{0})^{h(\sigma_{2})-h(\sigma_{0})-2h}\;g_{\sigmaseq}\big(\frac{x_{1}-x_{0}}{x_{2}-x_{0}}\big)$
then translate to the following
second order ODEs for $g_{\sigmaseq}(z)$:
\begin{align} \label{eq: ODE}
\begin{split}
0=\; & \kappa\;z^{2}(z-1)^{2}\;g_{\sigmaseq}''(z)+8\;z(z-1)(z-\frac{1}{2})\;g_{\sigmaseq}'(z)\\
 & +4\Big(z(z-2)h-z(z-1)h(\sigma_{2})+(z-1)h(\sigma_{0})\Big)\;g_{\sigmaseq}(z)
 \end{split} \\
\label{eq: ODE2}
 \begin{split}
 0 = \; &  \kappa\;z^{2}(z-1)^{2}\;g_{\sigmaseq}''(z) - 2 z (z-1) \Big( \kappa ( \tilde{\Delta} -1 ) (z-1)  + 2 (z-2) \Big) \; g'(z) \\
 & + \Big[  \Big(\kappa \tilde{\Delta}(\tilde{\Delta} - 1) +  4 \tilde{\Delta} -4 h(\sigma_0) \Big) (z-1)^2 - 4h \Big] \; g(z),
\end{split}
\end{align}
where we denoted the scaling exponent in $\widetilde{U}_{\sigmaseq}$ by $ \tilde{\Delta} = h(\sigma_{2})-h(\sigma_{0})-2h$. The two ODEs~\eqref{eq: ODE}~--~\eqref{eq: ODE2} coincide if $\sigma_0 = \sigma_2$. 
We analyze below separately the different $\sigmaseq$ allowed by the selection rules \eqref{eq: form of 3pt function with selection}, finding that the asymptotics
\begin{align} \label{eq: required asymptotics for the ODE for 4pt conformal block}
g_{\sigmaseq}(z)\sim\; & C_{\sigma_{0}}^{\sigma_{1}}C_{\sigma_{1}}^{\sigma_{2}}\;z^{h(\sigma_{1})-h(\sigma_{0})-h} \qquad \text{ as } z \to 0 
\end{align}
obtained from~\eqref{eq: asymptotics for x2}~--~\eqref{eq: asymptotics for x1} specify a unique solution to the first ODE~\eqref{eq: ODE} in all cases. These solutions are explicit, and can be verified to also satisfy the second ODE~\eqref{eq: ODE2}.


Denote $\sigma_{0}=\lambda$. 
By the selection rules, there are 
four possibilities when $n=2$, which we label as follows:~
\begin{center}
\begin{tabular}{c|cccc}
 &
 \textbf{up-wedge}  \qquad  &
 \textbf{down-wedge}  \qquad &
 \textbf{up-slope}  \qquad &
 \textbf{down-slope}  \qquad \\
\hline
$\sigmaseq$ &
 $(\lambda,\lambda+1,\lambda)$ &
 $(\lambda,\lambda-1,\lambda)$ &
 $(\lambda,\lambda+1,\lambda+2)$ &
 $(\lambda,\lambda-1,\lambda-2)$ \\
abbreviation &
 $\upwedge$ &
 $\downwedge$ &
 $\upslope$ &
 $\downslope$ \\
Figure &
~\ref{fig: up-wedge} &
~\ref{fig: down-wedge} &
~\ref{fig: up-slope} &
~\ref{fig: down-slope} \\
\end{tabular}
\smallskip
\end{center}
%
In the case of an up-slope, the only solution of the 
ODE~\eqref{eq: ODE} with the correct asymptotics~\eqref{eq: required asymptotics for the ODE for 4pt conformal block} is
\begin{align*}
g_{\upslope}(z)=\; & C_{\lambda}^{+}C_{\lambda+1}^{+}\times(1-z)^{\frac{2}{\kappa}}\times z^{h(\lambda+1)-h(\lambda)-h}.
\end{align*}
This also satisfies the ODE~\eqref{eq: ODE2}. A similar conclusion holds in the case of a down-slope:
\begin{align*}
g_{\downslope}(z)=\; & C_{\lambda}^{-}C_{\lambda-1}^{-}\times(1-z)^{\frac{2}{\kappa}}\times z^{h(\lambda-1)-h(\lambda)-h}.
\end{align*}
In the case of an up-wedge, the two ODEs coincide. The unique solution with the asymptotics~\eqref{eq: required asymptotics for the ODE for 4pt conformal block} is a slightly more complicated,
non-degenerate hypergeometric function
\begin{align*}
g_{\upwedge}(z)=\; & C_{\lambda}^{+}C_{\lambda+1}^{-}\,z^{\frac{2\lambda}{\kappa}}(1-z)^{\frac{\kappa-6}{\kappa}}\,\twoFone\Big(\frac{\kappa-4}{\kappa},\frac{4\lambda}{\kappa};\frac{4\lambda+4}{\kappa};z\Big).
\end{align*}
Similarly, in the case of a down-wedge, 
the solution is
\begin{align*}
g_{\downwedge}(z)=\; & C_{\lambda}^{-}C_{\lambda-1}^{+}\,z^{\frac{\kappa-2\lambda-4}{\kappa}}(1-z)^{\frac{\kappa-6}{\kappa}}\,\twoFone\Big(\frac{\kappa-4}{\kappa},\frac{2\kappa-8-4\lambda}{\kappa};\frac{2\kappa-4\lambda-4}{\kappa};z\Big).
\end{align*}

The asymptotics as $z\to1$ of such hypergeometric functions can be
obtained from the identities
$\twoFone(a,b;c;0)=1$ and~\cite[Equation~(15.3.6)]{AS-handbook}:
\begin{align*}
\twoFone(a,b;c;z)=\; & (1-z)^{c-a-b}\frac{\Gamma(c)\,\Gamma(a+b-c)}{\Gamma(a)\,\Gamma(b)}\twoFone(c-a,c-b;c-a-b+1;1-z)\\
 & +\frac{\Gamma(c)\,\Gamma(c-a-b)}{\Gamma(c-a)\,\Gamma(c-b)}\twoFone(a,b;a+b-c+1;1-z) .
\end{align*}
Because we assume $0<\kappa<8$, the parameters $a,b,c$
of the hypergeometric functions in both $g_{\upwedge}(z)$ and $g_{\downwedge}(z)$
satisfy $c-a-b=\frac{8-\kappa}{\kappa}>0$. Thus, in the limit $z\to1$
of the hypergeometric function, the first term above vanishes and the second
term tends to $\frac{\Gamma(c)\,\Gamma(c-a-b)}{\Gamma(c-a)\,\Gamma(c-b)}$.
This shows that we have
\begin{align*}
(1-z)^{\frac{6-\kappa}{\kappa}}\times g_{\upwedge}(z)\quad \longrightarrow \quad\; & C_{\lambda}^{+}C_{\lambda+1}^{-}\,\frac{\Gamma\big(\frac{4+4\lambda}{\kappa}\big)\,\Gamma\big(\frac{8-\kappa}{\kappa}\big)}{\Gamma\big(\frac{8-\kappa+4\lambda}{\kappa}\big)\,\Gamma\big(\frac{4}{\kappa}\big)}\\
(1-z)^{\frac{6-\kappa}{\kappa}}\times g_{\downwedge}(z)\quad \longrightarrow \quad\; & C_{\lambda}^{-}C_{\lambda-1}^{+}\,\frac{\Gamma\big(\frac{2\kappa-4-4\lambda}{\kappa}\big)\,\Gamma\big(\frac{8-\kappa}{\kappa}\big)}{\Gamma\big(\frac{4}{\kappa}\big)\,\Gamma\big(\frac{\kappa-4\lambda}{\kappa}\big)},
\end{align*}
and we can write down explicit asymptotics of the conformal
block functions $U_{\sigmaseq}(x_{1},x_{2})$ as 
$x_{1},x_{2} \to \xi$:
\begin{align}
\frac{U_{\upwedge}(x_{1},x_{2})}{(x_{2}-x_{1})^{-2h}}\quad \longrightarrow \quad\; & C_{\lambda}^{+}C_{\lambda+1}^{-}\,\frac{\Gamma\big(\frac{4+4\lambda}{\kappa}\big)\,\Gamma\big(\frac{8-\kappa}{\kappa}\big)}{\Gamma\big(\frac{8-\kappa+4\lambda}{\kappa}\big)\,\Gamma\big(\frac{4}{\kappa}\big)}\label{eq: non-normalized upwedge asymptotics}\\
\frac{U_{\downwedge}(x_{1},x_{2})}{(x_{2}-x_{1})^{-2h}}\quad \longrightarrow \quad\; & C_{\lambda}^{-}C_{\lambda-1}^{+}\,\frac{\Gamma\big(\frac{2\kappa-4-4\lambda}{\kappa}\big)\,\Gamma\big(\frac{8-\kappa}{\kappa}\big)}{\Gamma\big(\frac{4}{\kappa}\big)\,\Gamma\big(\frac{\kappa-4\lambda}{\kappa}\big)},\label{eq: non-normalized downwedge asymptotics}\\
\frac{U_{\upslope}(x_{1},x_{2})}{(x_{2}-x_{1})^{h(2)-2h}}\quad \longrightarrow \quad\; & C_{\lambda}^{+}C_{\lambda+1}^{+} \; 
\xi^{h(\lambda+2)-h(\lambda)-h(2)}\label{eq: non-normalized upslope asymptotics}\\
\frac{U_{\downslope}(x_{1},x_{2})}{(x_{2}-x_{1})^{h(2)-2h}}\quad \longrightarrow \quad\; & C_{\lambda}^{-}C_{\lambda-1}^{-} \; 
\xi^{h(\lambda-2)-h(\lambda)-h(2)}.\label{eq: non-normalized downslope asymptotics}
\end{align}
The last two expressions above,~\eqref{eq: non-normalized upslope asymptotics} and~\eqref{eq: non-normalized downslope asymptotics},
are proportional to one-point conformal block functions from $\Qrep_{\lambda}$
to $\Qrep_{\lambda\pm2}$ 
of a primary field operator of conformal
weight $h(2)=\frac{8-\kappa}{\kappa}$, whereas the first two,
\eqref{eq: non-normalized upwedge asymptotics} and~\eqref{eq: non-normalized downwedge asymptotics}, are
proportional to an one-point conformal block function from $\Qrep_{\lambda}$ to itself 
of
a primary field operator of conformal weight $h(0)=0$. 
The latter is
the identity operator, whose one-point conformal block function is just the constant $1$.
We now choose the
normalization constants $C_{\lambda}^{\pm}$ so that the coefficient of the identity
operator in~\eqref{eq: non-normalized upwedge asymptotics} 
equals one,
i.e., we set 
\begin{align*}
C_{\lambda}^{-} = \; & \frac{1}{C_{\lambda-1}^{+}}\times\frac{\Gamma\big(\frac{4-\kappa+4\lambda}{\kappa}\big)\,\Gamma\big(\frac{4}{\kappa}\big)}{\Gamma\big(\frac{4\lambda}{\kappa}\big)\,\Gamma\big(\frac{8-\kappa}{\kappa}\big)} \qquad
\text{for all $\lambda>0$.}
\end{align*}
Then, the coefficients $C_{\lambda}^{+}$ are the remaining free parameters.
The coefficient of the identity operator in~\eqref{eq: non-normalized downwedge asymptotics}
then becomes a ratio of gamma-functions, which can be further simplified to
\begin{align*}
C_{\lambda}^{-}C_{\lambda-1}^{+}\,\frac{\Gamma\big(\frac{2\kappa-4-4\lambda}{\kappa}\big)\,\Gamma\big(\frac{8-\kappa}{\kappa}\big)}{\Gamma\big(\frac{4}{\kappa}\big)\,\Gamma\big(\frac{\kappa-4\lambda}{\kappa}\big)}=\; & \frac{\Gamma\big(\frac{2\kappa-4-4\lambda}{\kappa}\big)\,\Gamma\big(\frac{4-\kappa+4\lambda}{\kappa}\big)}{\Gamma\big(\frac{\kappa-4\lambda}{\kappa}\big)\,\Gamma\big(\frac{4\lambda}{\kappa}\big)}\\
=\; & \frac{\sin\big(\pi\frac{4\lambda}{\kappa}\big)}{\sin\big(\pi(\frac{4\lambda+4}{\kappa}-1)\big)}=-\frac{\sin\big(\pi\frac{4\lambda}{\kappa}\big)}{\sin\big(\pi\frac{4\lambda+4}{\kappa}\big)} ,
\end{align*}
using the identity $\Gamma(w)\Gamma(1-w) = \pi / \sin(\pi w)$ twice.
Introducing the parameter $q=e^{\ii\pi4/\kappa}$ and $q$-integers
$\qnum n=\frac{q^{n}-q^{-n}}{q-q^{-1}}=\frac{\sin(4\pi n/\kappa)}{\sin(4\pi/\kappa)}$,
this takes the simple form
\begin{align*}
C_{\lambda}^{-}C_{\lambda-1}^{+}\,\frac{\Gamma\big(\frac{2\kappa-4-4\lambda}{\kappa}\big)\,\Gamma\big(\frac{8-\kappa}{\kappa}\big)}{\Gamma\big(\frac{4}{\kappa}\big)\,\Gamma\big(\frac{\kappa-4\lambda}{\kappa}\big)}=\; & -\frac{\qnum{\lambda}}{\qnum{\lambda+1}}.
\end{align*}
With the chosen normalization convention and when $0<\kappa<8$, the
leading asymptotics~\eqref{eq: non-normalized upwedge asymptotics}~--~\eqref{eq: non-normalized downslope asymptotics}
as $x_{1},x_{2}\to\xi\in(0,\infty)$ of the conformal block functions
can thus be summarized as
\begin{align*} 
\frac{U_{\sigmaseq}(x_{1},x_{2})}{(x_{2}-x_{1})^{-2h}}\quad \longrightarrow \quad\; & \begin{cases}
1 & \text{if }\sigmaseq=\upwedge\\
-\frac{\qnum{\lambda}}{\qnum{\lambda+1}} & \text{if }\sigmaseq=\downwedge\\
0 & \text{if }\sigmaseq=\upslope\text{ or }\sigmaseq=\downslope.
\end{cases}
\end{align*}

In the case of general $n$ and $\sigmaseq=(\sigma_{0},\sigma_{1},\ldots,\sigma_{n})$,
the leading asymptotics on pairwise diagonals can be inferred recursively from the above calculation. 
Specifically, we get 
\begin{align*}
\frac{U_{\sigmaseq}(x_{1},\ldots,x_{n})}{(x_{j+1}-x_{j})^{-2h}}\quad \longrightarrow \quad\; & \begin{cases}
U_{\hat{\sigmaseq}}(x_{1},\ldots,x_{j-1},x_{j+2},\ldots,x_{n}) & \text{if $\sigma_{j-1} = \sigma_{j+1} = \sigma_j -1$}\\
-\frac{\qnum{\sigma_{j}+1}}{\qnum{\sigma_{j}+2}}\times U_{\hat{\sigmaseq}}(x_{1},\ldots,x_{j-1},x_{j+2},\ldots,x_{n}) & \text{if $\sigma_{j-1} = \sigma_{j+1} = \sigma_j +1$}\\
0 & \text{if $\sigma_{j-1} \neq \sigma_{j+1}$} ,
\end{cases}
\end{align*}
as $x_{j},x_{j+1}\to\xi\in(x_{j-1},x_{j+2})$,
where we denote
$\hat{\sigmaseq} = (\sigma_0, \sigma_1, \ldots, \sigma_{j-1}, \sigma_{j+2}, \ldots, \sigma_{2N})$.

\subsection{Defining properties of conformal block functions}
\label{sub: defining conformal block functions}

So far in this section we have provided background on conformal block functions, so as to
have a self-contained justification of their properties that we use as their
definition 
in the rest of this article.
We 
only consider the conformal block functions that contribute to
vacuum expected values, in which case $\sigma_0 = 0$ and $\sigma_n = 0$,
and $n$ is necessarily even: $n=2N$. The sequence
$(\sigma_0 , \sigma_1 , \ldots, \sigma_n)$ then forms a Dyck path
of $n=2N$ steps, see Figure~\ref{fig: Dyck paths}.
Instead of the notation $\sigmaseq$, we use the notation $\alpha \in \DP_N$ for this Dyck path, and
\begin{align*}
\ConfBlockFun_\alpha (x_1 , \ldots , x_{2N})
\end{align*}
for the corresponding conformal block function.



We next list the defining properties of
$\big(\ConfBlockFun_\alpha\big)_{\alpha \in \DP}$ in the form that they will be used.
Denote $h = \frac{6-\kappa}{2\kappa}$ as before.
The required properties of $\ConfBlockFun_\alpha$ for $\alpha \in \DP_N$
are the partial differential equations
\begin{align}
\label{eq: PDE for conformal blocks} \tag{PDE} 
\left[ \frac{\kappa}{2} \pdder{x_j}
    + \sum_{\substack{i=1,\ldots,2N \\ i\neq j}} \Big( \frac{2}{x_i-x_j} \pder{x_i} - \frac{2 h}{(x_i-x_j)^2} \Big) \right] \; & \ConfBlockFun_\alpha (x_1 , \ldots, x_{2N}) = 0 \\
\nonumber
& \text{for all } j \in \set{1,\ldots,2N} ,
\end{align}
the M\"obius covariance
\begin{align}
\label{eq: COV for conformal blocks} \tag{COV} 
& \ConfBlockFun_\alpha(x_1 , \ldots, x_{2N}) = 
    \prod_{i=1}^{2 N} \Mob'(x_i)^{h} \times \ConfBlockFun_\alpha (\Mob(x_1) , \ldots, \Mob(x_{2N}))  \\
\nonumber
& \text{for all } \Mob(z) = \frac{a z + b}{c z + d}, \; \text{ with } a,b,c,d \in \bR, \; ad-bc > 0, 
 \text{ such that } \Mob(x_1) < \cdots < \Mob(x_{2N}) ,
\end{align}
and the 
recursive asymptotics properties
\begin{align}\label{eq: ASY for conformal blocks} \tag{$\ConfBlockFun$-ASY}
\lim_{x_j , x_{j+1} \to \xi} 
\frac{\ConfBlockFun_\alpha (x_1 , \ldots, x_{2N})}{(x_{j+1} - x_j)^{-2h}}
= \; & \begin{cases}
0 & \text{if } \slopeat{j} \in \alpha \\
   \ConfBlockFun_{\alpha \removeupwedge{j}} (x_1, \ldots, x_{j-1} , x_{j+2} , \ldots, x_{2N}) 
& \text{if } \upwedgeat{j} \in \alpha \\
-\frac{\qnum{\alpha(j)+1}}{\qnum{\alpha(j)+2}} \times
\ConfBlockFun_{\alpha \removedownwedge{j}} (x_1, \ldots, x_{j-1} , x_{j+2} , \ldots, x_{2N})  & \text{if } \downwedgeat{j} \in \alpha ,
    \end{cases}  \\
\nonumber 
& \text{for any $j \in \set{1, \ldots, 2N-1}$ and $\xi \in (x_{j-1}, x_{j+2})$,}
\end{align}
%
where 
the square bracket expressions are the $q$-integers
$\qnum{n} = \frac{q^n - q^{-n}}{q - q^{-1}}$
with the parameter $q = e^{\ii \pi 4 / \kappa}$ depending on $\kappa$.
Finally, the case $N=0$ fixes an overall normalization when
we require that $\ConfBlockFun_\emptywalk = 1$ for the Dyck path $\emptywalk \in \DP_0$.

In one of the main results of this article, Theorem~\ref{thm: change of basis theorem} in Section~\ref{sub: change of basis results}, 
we in particular prove that the conformal block functions $\ConfBlockFun_{\alpha}$ are uniquely determined by 
the properties above.

\bigskip

\section{Change of basis between conformal block functions and pure partition functions}
\label{sec: change of basis}



This section contains our first main result. It states first of all that
the conformal block functions and the multiple $\SLE$ pure partition functions, 
whose definition will be recalled below in Section~\ref{sub: pure partition functions}, both
form a basis of the same solution space of the system~\eqref{eq: PDE for conformal blocks} of partial differential equations.
Moreover, it gives an explicit $q$-combinatorial formula for the change of basis matrix.
%
%

We begin by discussing the space of functions.
Fix $N \in \bN$, and consider the system of $2N$ second order partial differential
equations 
and M\"obius covariance conditions 
as in Section~\ref{sub: defining conformal block functions},
\begin{align}
\label{eq: PDE for F} \tag{PDE} 
\left[ \frac{\kappa}{2} \pdder{x_j}
    + \sum_{\substack{i=1,\ldots,2N \\ i\neq j}} \left( \frac{2}{x_i-x_j} \pder{x_i} - \frac{2 h}{(x_i-x_j)^2} \right) \right] F (x_1 , \ldots, x_{2N}) = 0 & \\
\nonumber
\text{for all } j \in \set{1,\ldots,2N} , & \\
\label{eq: COV for F} \tag{COV} 
F (x_1 , \ldots, x_{2N}) = 
    \prod_{i=1}^{2 N} \Mob'(x_i)^{h} \times F (\Mob(x_1) , \ldots, \Mob(x_{2N})) & \\
\nonumber
\text{for all } \Mob(z) = \frac{a z + b}{c z + d}, \; \text{ with } a,b,c,d \in \bR, \; ad-bc > 0, 
    \text{ such that } \Mob(x_1) < \cdots < \Mob(x_{2N}) , &
\end{align}
for complex valued functions $F$ defined on the set
\begin{align*}
\chamber_{2N} := \; & \set{ (x_1, x_2, \ldots, x_{2N}) \in  \R^n \; \Big| \; x_1 < x_2 < \cdots < x_{2N} }
\end{align*}
of $2N$-tuples of real variables in increasing order. 
Require moreover
that $F$ has
at most polynomial growth on pairwise diagonals and at infinity in the sense that
\begin{align} 
\label{eq: power law growth bound} \tag{GROW}
\begin{split}
& \text{there exist positive constants $C,p > 0$ such that we have }  \\
& \big| F(x_1,\ldots,x_{2N}) \big| \leq C \times
\prod_{i < j} \max\left( (x_j - x_i)^p, (x_j - x_i)^{-p} \right)
\qquad \text{for all } (x_1,\ldots,x_{2N}) \in \chamber_{2N}.
\end{split}
\end{align}
We consider the following space of solutions:
\begin{align*} 
\Sol_N := \set{ F \colon \chamber_{2N} \to \bC \;\big|\; F \text{ satisfies 
\eqref{eq: PDE for F},~\eqref{eq: COV for F},
and~\eqref{eq: power law growth bound} } }.
\end{align*}
The dimension of this space is known to be the $N$:th Catalan number, $\dmn \Sol_N = \Catalan_N$,
and 
the multiple $\SLE$ pure partition functions
form a basis for $\Sol_N$,
as we 
recall precisely in Section~\ref{sub: pure partition functions}. 

\subsection{Multiple $\SLE$ pure partition functions}
\label{sub: pure partition functions}
In many situations in planar statistical physics, boundary conditions force the
existence of multiple macroscopic interfaces. In the scaling limit at criticality, such
interfaces are described by multiple $\SLEk$ curves 
with $\kappa$ depending on the model. Contrary to, e.g., a single chordal $\SLEk$ curve,
the law of a multiple $\SLEk$ with fixed 
number $N$ of curves
is not unique~--- instead, the possible laws form a non-trivial convex set.
It is thus natural to express any multiple $\SLEk$ as a convex combination of
the extremal points of this convex set: the pure geometries,
in which the curves connect the starting points of the interfaces  
pairwise in a deterministic planar pair partition. 
The pure geometries are thus
indexed by planar pair partitions, which in turn are in bijection with Dyck paths.
For background on multiple $\SLE$s, we refer 
to~\cite{BBK-multiple_SLEs,Dubedat-commutation,LK-configurational_measure,KP-pure_partition_functions_of_multiple_SLEs},
and for results on their role as scaling limits,
to~\cite{CS-universality_in_2d_Ising, Izyurov-critical_Ising_interfaces_in_multiply_connected_domains,
KKP-boundary_correlations_in_planar_LERW_and_UST,
PH-Global_multiple_SLEs_and_pure_partition_functions, Wu-convergence_of_Ising_interfaces_to_hypergeometric_SLE,
Beffara_Peltola_Wu-Uniqueness_of_global_multiple_SLEs, KS-configurations_of_FK_interfaces}.

For the purposes of this article, the important aspect of
multiple $\SLE$s is their partition functions $\PartF$, which essentially define the
multiple $\SLEk$ by giving the Radon-Nikodym density of its law with respect to independent chordal
$\SLEk$ laws, see \cite{Dubedat-commutation, KP-pure_partition_functions_of_multiple_SLEs}.
In particular, each pure geometry with connectivity encoded by a Dyck path $\alpha$ has a
partition function denoted by $\PartF_\alpha$.
%
These functions satisfy the partial differential equations~\eqref{eq: PDE for F} and
M\"obius covariance~\eqref{eq: COV for F} as before, and the following recursive asymptotics:
\begin{align}
\label{eq: ASY for multiple SLEs} \tag{$\PartF$-ASY} 
\lim_{x_j , x_{j+1} \to \xi} 
\frac{\PartF_\alpha(x_1 , \ldots, x_{2N})}{(x_{j+1} - x_j)^{-2h}}
= & \begin{cases}
    \PartF_{\alpha \removeupwedge{j} } (x_1, \ldots, x_{j-1} , x_{j+2} , \ldots, x_{2N}) & \text{ if } \upwedgeat{j} \in \alpha \\
    0 & \text{ if } \upwedgeat{j} \notin \alpha
    \end{cases} \\
\nonumber
& \text{for any $j \in \set{1, \ldots, 2N-1}$, and $\xi \in (x_{j-1}, x_{j+2})$.}
\end{align}
As stated in the following proposition, these requirements together with the
normalization condition $\PartF_\emptywalk = 1$ uniquely determine the functions $\PartF_\alpha$,
called the \emph{multiple SLE pure partition functions}.
%
%

\begin{prop}
\label{prop: solution space with power law bound}
Let $\kappa \in (0,8) \setminus \bQ$.
There exists a unique collection of functions $\big(\PartF_\alpha\big)_{\alpha \in \DP}$,
such that $\PartF_\alpha \in \Sol_N$ when $\alpha \in \DP_N$, $\PartF_\emptywalk = 1$,
and~\eqref{eq: ASY for multiple SLEs} holds for all $\alpha$.
Moreover, for any $N \in \bZnn$, the functions $\big( \PartF_\alpha \big)_{\alpha \in \DP_N}$
form a basis of the 
solution space $\Sol_N$.
\end{prop}
\begin{proof}
By~\cite[Theorem~8]{FK-solution_space_for_a_system_of_null_state_PDEs_3}, 
we have
$\dmn \Sol_N = \Catalan_N$. 
On the other hand,~\cite[Theorem~4.1]{KP-pure_partition_functions_of_multiple_SLEs}
shows that the pure partition functions $\big( \PartF_\alpha \big)_{\alpha \in \DP_N}$
form a linearly independent set in this space
(the power-law bound~\eqref{eq: power law growth bound} can be verified from the explicit form of 
the functions as Coulomb gas integrals, 
see~\cite{KP-pure_partition_functions_of_multiple_SLEs, KP-conformally_covariant_boundary_correlation_functions_with_a_quantum_group}). 
The assertion follows, since $\# \DP_N = \Catalan_N$.
\end{proof}

\subsection{The change of basis result}
\label{sub: change of basis results}

We now show how to express the conformal block functions $\ConfBlockFun_\alpha$ in the basis of
the multiple $\SLE$ pure partition functions $\PartF_\alpha$ using weighted incidence
matrices of the parenthesis reversal relation. From this, it follows that the conformal
block functions are well-defined and also form a basis.

We take the weight of a Dyck tile $t$ at height $h_t$ to be
\begin{align}\label{eq: weights in terms of q-numbers}
w(t) := \frac{\qnum{h_t}}{\qnum{h_t+1}} ,
\end{align}
where $\qnum{n} = \frac{q^n - q^{-n}}{q - q^{-1}}$ and $q = e^{\ii \pi 4/\kappa}$ as before.
Denote by $\genMmat = (\genMmat_{\alpha, \beta})$ the correspondingly weighted incidence 
matrix~\eqref{eq: def of weighted incidence matrix}: its non-zero elements are
\begin{align} \label{eq: M}
\genMmat_{\alpha, \beta} = \prod_{t \in \nestedtilingof(\alpha / \beta)} (-w(t)) , \qquad \text{ for $\alpha \KWleq \beta$,}
\end{align}
where $\nestedtilingof(\alpha / \beta)$ is the nested Dyck tiling of the skew Young diagram 
$\alpha/\beta$.
Proposition~\ref{prop: weighted KW incidence matrix inversion} shows that the matrix 
$\genMmat$ is invertible and the non-zero matrix elements of its inverse are 
\begin{align}  \label{eq: Minv}
\genMinv_{\alpha, \beta} = \sum_{T \in \CItilingsof (\alpha/\beta)} \prod_{t \in T} w(t) , \qquad \text{ for $\alpha \DPleq \beta$,}
\end{align}
with $\CItilingsof (\alpha/\beta)$ the family of cover-inclusive Dyck tilings of 
the skew Young diagram $\alpha/\beta$.
Examples of these matrices are 
shown in Figures~\ref{fig: quantum KW matrices only}~and~\ref{fig: quantum CIDT matrices only}.

\begin{thm}\label{thm: change of basis theorem}
There exists a unique collection $\big(\ConfBlockFun_\alpha \big)_{\alpha \in \DP}$ 
such that $\ConfBlockFun_\alpha \in \Sol_N$ when $\alpha \in \DP_N$,
$\ConfBlockFun_\emptywalk = 1$, and 
the asymptotics~\eqref{eq: ASY for conformal blocks} hold. 
For any $\alpha \in \DP_N$, the function $\ConfBlockFun_\alpha$ 
of this collection can be written in the basis $\big( \PartF_\beta \big)_{\beta \in \DP_N}$ of
Proposition~\ref{prop: solution space with power law bound} as
\begin{align*}
\ConfBlockFun_\alpha =
\sum_{\beta \in \DP_N} \genMmat_{\alpha,\beta}\; \PartF_\beta,
\end{align*}
where $\genMmat$ is the weighted incidence matrix of the parenthesis reversal
relation with weights~\eqref{eq: weights in terms of q-numbers}.
Moreover, for any $N \in \bZnn$, the functions $\big( \ConfBlockFun_\alpha \big)_{\alpha \in \DP_N}$
form a basis of the 
solution space $\Sol_N$ and
\begin{align*}
\PartF_\alpha =
\sum_{\beta \in \DP_N} \genMinv_{\alpha,\beta}\;\ConfBlockFun_\beta .
\end{align*}
\end{thm}

%

\begin{rem}
The above change of basis formulas 
can be regarded as analogues of Fomin's formulas in the following sense.
In a special planar case, general Fomin's formulas~\cite{Fomin-LERW_and_total_positivity} yield
linear relationships between determinants of discrete Green's functions and 
probabilities of certain planar connectivity events for the uniform spanning tree~\cite[Section~3.4]{KKP-boundary_correlations_in_planar_LERW_and_UST}.
This linear system is encoded in an invertible matrix $(\Mmat_{\alpha, \beta})$ indexed by Dyck paths, 
whose non-zero entries are of the form~\eqref{eq: M} with unit tile weights $w=1$, 
instead of the $q$-dependent weights~\eqref{eq: weights in terms of q-numbers}.
Thus, the structure of the non-zero entries is encoded in the same combinatorial relation $\alpha \KWleq \beta$. 
In the scaling limit as the mesh of the graph tends to zero, these planar connectivity probabilities tend to
the pure partition functions of multiple $\SLE_\kappa$ for $\kappa = 2$,
and the determinants tend to a distinguished basis of the solution space $\Sol_N$ with $\kappa = 2$,
closely related to the conformal blocks
--- 
see~\cite[Section~5]{KW-boundary_partitions_in_trees_and_dimers}
and~\cite[Theorems~3.12~and~4.1~and~Proposition~4.6]{KKP-boundary_correlations_in_planar_LERW_and_UST}. 
\end{rem}

\begin{proof}[Proof of Theorem~\ref{thm: change of basis theorem}]
Let $\big( \ConfBlockFun_\alpha \big)_{\alpha \in \DP}$ be any collection of functions such that
$\ConfBlockFun_\alpha \in \Sol_N$ for $\alpha \in \DP_N$
and the asymptotics~\eqref{eq: ASY for conformal blocks} hold with $\ConfBlockFun_\emptywalk = 1$.
Write $\ConfBlockFun_\alpha$ in the basis $\big( \PartF_\beta \big)_{\beta \in \DP_N}$ of $\Sol_N$ as
\begin{align*}
\ConfBlockFun_\alpha = \sum_{\beta \in \DP_N} M^{(N)}_{\alpha, \beta} \; \PartF_\beta ,
\end{align*}
which, for each $N \in \bN$, defines a matrix $M^{(N)} \in \bC^{\DP_N \times \DP_N}$.
The 
recursive asymptotics~\eqref{eq: ASY for conformal blocks} then
equivalently require that the matrices $M^{(N)}_{\alpha, \beta}$ satisfy the initial condition $M^{(0)}=1$ and the recursion
\begin{align}
\begin{split}
& \lim_{x_j , x_{j+1} \to \xi} \frac{1}{(x_{j+1} - x_j)^{-2h}}
    \sum_{\beta \in \DP_N} M^{(N)}_{\alpha, \beta} \; \PartF_\beta (x_1 , \ldots, x_{2N}) \\
= \; & \begin{cases} 
0 & \text{if } \slopeat{j} \in \alpha \\
\label{eq: coblo walk projection conditions 2}
 \sum_{\hat{\beta} \in \DP_{N-1}} M^{(N-1)}_{\hat{\alpha}, \hat{\beta} } \; \PartF_{ \hat{\beta} } (x_1 , \ldots, x_{j-1} , x_{j+2}, \ldots, x_{2N})
& \text{if } \upwedgeat{j} \in \alpha \\
- \frac{\qnum{\alpha(j)+1}}{\qnum{\alpha(j)+2}} \times
 \sum_{\hat{\beta} \in \DP_{N-1}} M^{(N-1)}_{\hat{\alpha}, \hat{\beta} } \; \PartF_{ \hat{\beta} } (x_1 , \ldots, x_{j-1} , x_{j+2}, \ldots, x_{2N})
& \text{if } \downwedgeat{j} \in \alpha
\end{cases}
\end{split}
\end{align}
where we denote $\alpha \removewedge{j} = \hat{\alpha}$.
Now recall the asymptotics properties~\eqref{eq: ASY for multiple SLEs} of pure partition functions,
and note that each $\hat{\beta} \in \DP_{N-1}$ determines a unique 
$\beta \in \DP_N$ with $\upwedgeat{j} \in \beta$ such that $\beta \removeupwedge{j} = \hat{\beta}$.
The left-hand side of~\eqref{eq: coblo walk projection conditions 2} then becomes
\begin{align*}
& \lim_{x_j , x_{j+1} \to \xi} \frac{1}{(x_{j+1} - x_j)^{-2h}}
    \sum_{\beta \in \DP_N} M^{(N)}_{\alpha, \beta} \; \PartF_\beta (x_1 , \ldots, x_{2N}) \\
= \; & \sum_{\hat{\beta} \in \DP_{N-1}} M^{(N)}_{\alpha, \beta} \; \PartF_{ \hat{\beta} } (x_1 , \ldots, x_{j-1} , x_{j+2}, \ldots, x_{2N}) .
\end{align*}
Since $\big( \PartF_{ \hat{\beta} } \big)_{\hat{\beta}  \in \DP_{N - 1}}$ is a basis,
the recursion~\eqref{eq: coblo walk projection conditions 2} is equivalent to
the following: 
for any $j \in \set{1,\ldots,2N-1}$ and any $\beta \in \DP_N$ such that $\upwedgeat{j} \in \beta$, we have
\begin{align}\label{eq: recursion for the change of basis matrix alex}
M_{\alpha,\beta}^{(N)} = \begin{cases} 
   0 & \text{if } \slopeat{j} \in \alpha \\
   M_{\hat{\alpha},\hat{\beta}}^{(N-1)}
   & \text{if } \upwedgeat{j} \in \alpha \\
   -\frac{\qnum{\alpha(j)+1}}{\qnum{\alpha(j)+2}} \times
   M_{\hat{\alpha},\hat{\beta}}^{(N-1)} 
   & \text{if } \downwedgeat{j} \in \alpha,
\end{cases} 
\end{align}
where we denote by $\hat{\alpha} = \alpha \removewedge{j} \in \DP_{N-1}$ and 
$\hat{\beta} = \beta \removeupwedge{j} \in \DP_{N-1}$. Finally,
Proposition~\ref{prop: recursion for matrix elements} states that the
recursion~\eqref{eq: recursion for the change of basis matrix alex} holds if
and only if the matrices  $M^{(N)}$ are, for any $N$, 
the weighted incidence matrices of the parenthesis reversal relation, $M=\genMmat$. The rest follows since
the matrix $\genMmat$ is, for any $N$, invertible by Proposition~\ref{prop: weighted KW incidence matrix inversion}.
\end{proof}

\bigskip

\section{Direct construction of conformal block functions by a quantum group method}
\label{sec: construction of conformal block functions}
In the preceding section, we expressed the conformal block function~$\ConfBlockFun_\alpha$
as linear combinations of multiple $\SLE$ pure partition functions $\PartF_\alpha$ and vice versa, 
generalizing Fomin's formula~\cite{Fomin-LERW_and_total_positivity, KKP-boundary_correlations_in_planar_LERW_and_UST}.
These expressions can also be viewed as a construction of the conformal block functions, which
however relies on an earlier construction of the multiple SLE pure partition
functions and detailed information about the solution
space~\cite{FK-solution_space_for_a_system_of_null_state_PDEs_1,
FK-solution_space_for_a_system_of_null_state_PDEs_2, FK-solution_space_for_a_system_of_null_state_PDEs_3,
KP-pure_partition_functions_of_multiple_SLEs}.
In this section, we provide an alternative, more direct construction of the
conformal block functions $\ConfBlockFun_\alpha$ based on a quantum group method developed
in~\cite{KP-conformally_covariant_boundary_correlation_functions_with_a_quantum_group}.
In analogy with the core underlying idea of
conformal blocks as discussed in Section~\ref{sec: conformal block functions background},
the present construction employs the Dyck path $\alpha$
as labeling a sequence of representations of the quantum group $\Uqsltwo$.
This quantum group construction furthermore sheds some light on why $q$-combinatorial
formulas for conformal blocks appear in the first place.

Generalizations of this construction for
conformal blocks in representations of the quantum group $\Uqsltwo$ and in
relation to CFT correlation functions are used and studied in~\cite{Flores-Peltola:WJTL_algebra, Flores-Peltola:Colored_braid_representations_and_QSW,
Flores-Peltola:Monodromy_invariant_correlation_function}.

\subsection{\label{sub: quantum group}The quantum group and its representations}

We begin by introducing the needed definitions and notation
about the quantum group $\Uqsltwo$ and its representations.
For more background, see, e.g.,~\cite{Kassel-Quantum_groups} and references therein.
Let $q = e^{\ii \pi 4 / \kappa}$ as before.
As a $\bC$-algebra, $\Uqsltwo$ is generated by the elements $K,K^{-1},E,F$ subject to 
the relations
\begin{align} \label{eq: quantum group relations}
\begin{split}
 KK^{-1}=&\; 1=K^{-1}K,\qquad KE=q^{2}EK,\qquad KF=q^{-2}FK,\\
 EF-FE=&\; \frac{1}{q-q^{-1}}\left(K-K^{-1}\right) . 
\end{split}
\end{align}
It has a Hopf algebra structure, with coproduct $\Delta \colon \Uqsltwo \rightarrow \Uqsltwo\tens\Uqsltwo$ given on its generators~by
\begin{align}\label{eq: coproduct} 
\Hcp(E) = \; & E\tens K + 1\tens E,\qquad
\Hcp(K) = \; K \tens K,\qquad
\Hcp(F) = \; F \tens 1 + K^{-1} \tens F . 
\end{align}
The coproduct is used to define the action of 
the Hopf algebra $\Uqsltwo$
on tensor products of representations as follows. If the coproduct of an element $X \in \Uqsltwo$ reads
\[ \Hcp(X) = \sum_i X_i' \tens X_i'' , \]
and if $V'$ and $V''$ are two representations, then
$X$ acts on a tensor $v' \tens v'' \in V' \tens V''$ by the formula
\begin{align*}
X.(v' \tens v'') = \; & \sum_i X_i'.v' \tens X_i''.v'' .
\end{align*}
Tensor product representations with 
$n$ tensor components are defined using the $(n-1)$-fold coproduct
\begin{align*}
\Hcp^{(n)} = \; & (\Hcp\tens\id^{\tens(n-2)})\circ(\Hcp\tens\id^{\tens(n-3)})\circ\cdots
\circ(\Hcp\tens\id)\circ\Hcp ,
\qquad \Hcp^{(n)} \colon\Uqsltwo\rightarrow\Big(\Uqsltwo\Big)^{\tens n} .
\end{align*}
By the coassociativity property 
$(\id\tens\Hcp)\circ\Hcp = (\Hcp\tens\id)\circ\Hcp$ 
the tensor products of representations thus defined are associative, i.e.,
there is no need to specify the order in which the tensor products are formed.

For each $d \in \bN$, the quantum group
$\Uqsltwo$ has an irreducible representation $\Wd_d$ of dimension $d$,
obtained by suitably $q$-deforming the $d$-dimensional irreducible representation
of the simple Lie algebra $\mathfrak{sl}_2$. Of primary importance to
us is the two-dimensional irreducible representation $\Wd_2$: it has a basis $\set{\Wbas_0 , \Wbas_1}$
on which the generators act by
\begin{align*} 
K. \Wbas_0 = q \, \Wbas_0 , \qquad
K. \Wbas_1 = q^{-1} \, \Wbas_1 , \qquad
E. \Wbas_0 = 0 , \qquad
E. \Wbas_1 = \Wbas_0 , \qquad
F. \Wbas_0 = \Wbas_1 , \qquad
F. \Wbas_1 = 0 .
\end{align*}
A similar explicit definition of the $d$-dimensional irreducible
$\Wd_d$ can be found in, e.g.,~\cite{KP-conformally_covariant_boundary_correlation_functions_with_a_quantum_group}.
The tensor product of two two-dimensional irreducibles
decomposes as a direct sum of subrepresentations,
\[ \Wd_2 \tens \Wd_2 \isom \Wd_1 \oplus \Wd_3 , \]
where $\Wd_1$ is a one-dimensional
subrepresentation spanned by the vector
\begin{align} \label{eq: singlet basis vector}
\Sbas 
= \frac{1}{q-q^{-1}}\left(\Wbas_{1}\tens\Wbas_{0}-q\,\Wbas_{0}\tens\Wbas_{1}\right),
\end{align}
and $\Wd_3$ is a three-dimensional irreducible subrepresentation with basis
\begin{align*}
\TRbas_+ 
= \Wbas_{0}\tens\Wbas_{0},\qquad 
\TRbas_0 
= q^{-1}\,\Wbas_{0}\tens\Wbas_{1} + \Wbas_{1}\tens\Wbas_{0},\qquad 
\TRbas_- 
= [2]\,\Wbas_{1}\tens\Wbas_{1} .
\end{align*}
We denote the projection onto the 
one-dimensional
subrepresentation
$\Wd_1 \subset \Wd_2 \tens \Wd_2$ 
by
\begin{align*}
\pi\colon\; &\Wd_{2}\tens\Wd_{2}\to\Wd_{2}\tens\Wd_{2} , \qquad 
\pi(\Sbas)=\Sbas \qquad \text{ and } \qquad
\pi(\TRbas_{+}) = \pi(\TRbas_{0}) = \pi(\TRbas_{-}) = 0 .
\end{align*}
The one-dimensional representation $\Wd_1$ is trivial in the sense that
it is the neutral element for tensor products of representations: for any representation $V$,
we have $\Wd_1 \tens V \isom V \isom V \tens \Wd_1$, and
$\Wd_1$ can thus simply be identified with 
the scalars $\bC$.
Using the identification $\Sbas \mapsto 1 \in \bC$, we denote the
projection from $\Wd_2 \tens \Wd_2$ to $\Wd_1 \isom \bC$ by
\begin{align}\label{eq: projection to singlet}
\hat{\pi}\colon\; &\Wd_{2}\tens\Wd_{2}\to \bC , \qquad
\hat{\pi}(\Sbas) = 1 \qquad \text{ and } \qquad
\hat{\pi}(\TRbas_{+}) = \hat{\pi}(\TRbas_{0}) = \hat{\pi}(\TRbas_{-}) = 0 .
\end{align}

More generally, we have the $q$-Clebsch-Gordan formula
\begin{align}\label{eq: decomposition of tensor product}
\Wd_{d_{2}}\tens\Wd_{d_{1}}\isom\; &
    \Wd_{d_{1}+d_{2}-1}\oplus\Wd_{d_{1}+d_{2}-3}\oplus\cdots\oplus\Wd_{|d_{1}-d_{2}|+3}\oplus\Wd_{|d_{1}-d_{2}|+1} 
\end{align}
for the direct sum decomposition of the tensor product of 
the irreducible representations of dimensions $d_1$ and $d_2$,
see, e.g.~\cite[Lemma~2.4]{KP-conformally_covariant_boundary_correlation_functions_with_a_quantum_group}.
Repeated application of the 
decomposition~\eqref{eq: decomposition of tensor product} gives
\begin{align}\label{eq: tensor power of two dimensionals}
\Wd_2^{\tens n} \; \isom \; \bigoplus_d m_d^{(n)} \, \Wd_d,
\end{align}
where the irreducible $\Wd_d$ of dimension $d$ appears with multiplicity
$m_d^{(n)}$, see, e.g.~\cite[Lemma~2.2]{KP-pure_partition_functions_of_multiple_SLEs}.
When $n = 2N$,
the trivial subrepresentation 
\begin{align} \label{eq: trivial submodule}
\HWsp_{2N}^{(0)}
:= \set{v \in \Wd_{2}^{\tens 2N}\;\Big|\;E.v=0 ,\;K.v=v}
\end{align}
coincides with the sum of all 
copies of $\Wd_1$, and has dimension equal to a Catalan number
\[ \dmn \HWsp_{2N}^{(0)} = m_1^{(2N)} = \Catalan_N . \]

Finally,
in the tensor product $\Wd_2^{\tens n}$, we denote by $\pi_j$ and $\hat{\pi}_j$
the projections $\pi$ and $\hat{\pi}$ 
acting on the $j$:th and $(j+1)$:st tensor components counting from the right, i.e.,
\begin{align*}
\pi_j := \; & \id^{\tens (n-1-j)} \tens \pi \tens \id^{\tens (j-1)} \colon\; \Wd_{2}^{\tens n} \to \; \Wd_{2}^{\tens n} \\
\hat{\pi}_j := \; & \id^{\tens (n-1-j)} \tens \hat{\pi} \tens \id^{\tens (j-1)} \colon\; \Wd_{2}^{\tens n} \to \; \Wd_{2}^{\tens (n-2)} .
\end{align*}

\subsection{Constructing conformal blocks}
The purpose of this section is to give a construction of the conformal block functions.
Our construction relies on the method 
introduced in~\cite{KP-conformally_covariant_boundary_correlation_functions_with_a_quantum_group},
called ``spin chain~-~Coulomb gas correspondence'', 
which is allows to solve conformal field theory PDEs with given boundary conditions
by quantum group calculations. We 
use the correspondence in the following form, which combines a special case of a more general theorem
in~\cite{KP-conformally_covariant_boundary_correlation_functions_with_a_quantum_group}
with additional information available in that special case~\cite{KP-pure_partition_functions_of_multiple_SLEs}.
\begin{prop} 
\label{prop: SCCG correspondence map}
Let $\kappa \in (0,8) \setminus \bQ$ and $q = e^{\ii \pi 4 / \kappa}$.
There exist explicit linear isomorphisms
\[ \sF \colon \HWsp_{2N}^{(0)} \to \Sol_N ,\] 
for all $N \in \bZnn$, with the following property.
Let $v \in \HWsp_{2N}^{(0)}$, 
and $j \in \set{1,2,\ldots,2N-1}$, and denote
$\hat{v} = \hat{\pi}_{j}(v) \in \HWsp_{2(N-1)}^{(0)}$.
Then, for any  $\xi \in (x_{j-1},x_{j+2})$, the function $\sF[v] \colon \chamber_{2N} \to \bC$ has the asymptotics 
\begin{align*}
\lim_{x_{j},x_{j+1}\to\xi}
\frac{\sF[v](x_{1},\ldots,x_{2N})}{(x_{j+1}-x_{j})^{-2 h}}
\; = \; \; & 
B \times 
\sF[\hat{v}](x_{1},\ldots,x_{j-1},x_{j+2},\ldots,x_{2N}) ,
\end{align*}
where $B = \frac{\Gamma(1-4/\kappa)^2}{\Gamma(2-8/\kappa)}$.
\end{prop}
\begin{proof}
Such a map $\sF$
was constructed in~\cite{KP-conformally_covariant_boundary_correlation_functions_with_a_quantum_group} 
and it follows from the explicit expressions of the functions $\sF[v]$ as Coulomb gas integrals
that $\sF[v] \in \Sol_N$ for all $v \in \HWsp_{2N}^{(0)}$.
That $\sF$ is injective is proven in~\cite{KP-pure_partition_functions_of_multiple_SLEs},
and comparison of dimensions then shows that $\sF$ is a linear isomorphism.
Finally, the asymptotics property follows immediately 
from~\cite[Theorem~4.17(ASY)]{KP-conformally_covariant_boundary_correlation_functions_with_a_quantum_group}.
\end{proof}

With the help of the correspondence $\sF$ of Proposition~\ref{prop: SCCG correspondence map},
the task of constructing the conformal block functions is reduced to the task of
constructing suitable vectors in the trivial subrepresentation $\HWsp_{2N}^{(0)}$ 
of a tensor product $\Wd_2^{\tens 2N}$ of
two-dimensional irreducible representations of the quantum group. This is achieved in the following 
proposition, which we prove in the end of this section.

\begin{prop}\label{prop: conformal block vectors}
There exists a unique collection of vectors $( \Coblobastwodim_\alpha )_{\alpha \in \DP}$, 
with $\Coblobastwodim_\alpha \in \HWsp_{2N}^{(0)}$ 
when $\alpha \in \DP_N$,
such that $\Coblobastwodim_\emptywalk = 1$ and the following projection properties hold:
\begin{align}
& \hat{\pi}_j(\Coblobastwodim_\alpha) =
\label{eq: coblo walk projection conditions}
\begin{cases} 
0 & \text{if } \slopeat{j} \in \alpha \\
\Coblobastwodim_{\alpha\removeupwedge{j}} 
& \text{if } \upwedgeat{j} \in \alpha \\
- \frac{\qnum{\alpha(j)+1}}{\qnum{\alpha(j)+2}} \times
\Coblobastwodim_{\alpha \removedownwedge{j}} 
& \text{if } \downwedgeat{j} \in \alpha
\end{cases}
\qquad\text{for all } j \in \set{1,\ldots,2N-1}.
\end{align}
Moreover, for any $N \in \bZnn$, the collection $( \Coblobastwodim_\alpha )_{\alpha \in \DP_N}$ is a basis of $\HWsp_{2N}^{(0)}$.
\end{prop}

Once Proposition~\ref{prop: conformal block vectors} is established, the
construction 
is immediate:
\begin{thm}\label{thm: conformal block functions}
Let $( \Coblobastwodim_\alpha )_{\alpha \in \DP}$ be the collection
of vectors in Proposition~\ref{prop: conformal block vectors}, and let 
$\sF \colon \HWsp_{2N}^{(0)} \to \Sol_N$ 
be the linear isomorphisms of Proposition~\ref{prop: SCCG correspondence map}.
Then the functions
\[ \ConfBlockFun_\alpha := \frac{1}{B^N} \times \sF[\Coblobastwodim_\alpha]  , \qquad \text{for $\alpha \in \DP_N$,}\]
satisfy the defining 
properties~\eqref{eq: PDE for conformal blocks}, \eqref{eq: COV for conformal blocks}, and \eqref{eq: ASY for conformal blocks}
of conformal block functions.
\end{thm}
\begin{proof}
This follows by combining Propositions~\ref{prop: conformal block vectors}
and~\ref{prop: SCCG correspondence map}.
\end{proof}

\subsection{\textbf{Proof of Proposition~\ref{prop: conformal block vectors}}}

The rest of this section constitutes the proof of Proposition~\ref{prop: conformal block vectors},
divided in four parts: uniqueness, construction, linear independence, and verification of
projection properties. 
The uniqueness is routine by considering the corresponding homogeneous problem.
The explicit construction of the vectors $\Coblobastwodim_\alpha$ is the essence of the proof.
In the construction, each Dyck path $\alpha \in \DP_N$
specifies a sequence of representations of the quantum group,
and we recursively assemble the vector $\Coblobastwodim_\alpha$ proceeding along this sequence.
Linear independence is transparent in the construction. Finally, for the projection properties
we just have to inspect a number of cases explicitly.



\subsubsection{\textbf{Uniqueness}}

Uniqueness of the collection
$( \Coblobastwodim_\alpha )_{\alpha \in \DP}$ 
of vectors satisfying the projection properties~\eqref{eq: coblo walk projection conditions}
follows from arguments exploited in similar contexts in
the articles~\cite{KP-pure_partition_functions_of_multiple_SLEs,
P-basis_for_solutions_of_BSA_PDEs}.
The crucial observation is the following lemma about the homogeneous problem.

\begin{lem}
\label{lem: all projections vanish}
If a vector $v \in \HWsp_{2N}^{(0)}$ 
satisfies the property
$\hat{\pi}_{j}^{(1)}(v) = 0$ for all $j \in \set{1,\ldots,2N-1}$, then $v = 0$.  
\end{lem}

\begin{proof}
See, e.g.,~\cite[Corollary~2.5]{KP-pure_partition_functions_of_multiple_SLEs}. 
\end{proof}

As a corollary, the solution space of the recursive projection
properties~\eqref{eq: coblo walk projection conditions} is one-dimensional,
with initial condition $u_\emptywalk \in \Wd_2^{\tens 0} \isom \bC$ determining
the solution.
\begin{cor}\label{cor: uniqueness}
Let $( \Coblobastwodim_\alpha )_{\alpha \in \DP}$ and 
$( \Coblobastwodim'_\alpha )_{\alpha \in \DP}$ 
be two collections of vectors 
$\Coblobastwodim_\alpha,\Coblobastwodim'_\alpha \in \HWsp_{2N}^{(0)}$
satisfying the projection properties~\eqref{eq: coblo walk projection conditions}, 
and having same initial condition $\Coblobastwodim'_\emptywalk = \Coblobastwodim_\emptywalk$. 
Then we have
\begin{align*}
\Coblobastwodim'_\alpha = \Coblobastwodim_\alpha\qquad\text{for all }\alpha \in \DP.
\end{align*}
\end{cor}
\begin{proof}
Let $N \geq 1$ and suppose the condition
$\Coblobastwodim'_\beta = \Coblobastwodim_\beta$ holds 
for all $\beta \in \DP_{N-1}$.
Then, for any $\alpha \in \DP_N$, 
the difference $v = \Coblobastwodim'_\alpha - \Coblobastwodim_\alpha$
satisfies $\hat{\pi}_j(v) = 0$ for all $j \in \set{1,\ldots,2N-1}$, 
so $v = 0$ by Lemma~\ref{lem: all projections vanish}. 
The assertion follows by induction on $N$.
\end{proof}

\subsubsection{\textbf{Construction}}

We now construct the vectors $\Coblobastwodim_\alpha$ of Proposition~\ref{prop: conformal block vectors} and show that they lie in the correct subspaces $\HWsp_{2N}^{(0)}$.
In the intermediate steps of the construction,
we encounter vectors in the highest weight vector 
spaces
\begin{align} \label{eq: highest vector space}
\HWsp_{n}^{(\nodef)} = \set{ v \in \Wd_{2}^{\tens n} \; \Big| \; E.v=0 ,\;K.v = q^{\nodef}v } . 
\end{align}
These spaces consist of generators of the $\Wd_{d}$-isotypic components in the tensor
product~\eqref{eq: tensor power of two dimensionals} with $d = \nodef + 1$:
for any non-zero $v \in \HWsp_{n}^{(\nodef)}$, the collection $(F^k.v)_{k=0}^s$ 
obtained from $v$ by the action of the generator $F$ spans a subrepresentation 
isomorphic to $\Wd_{d}$ in $\Wd_{2}^{\tens n}$.
For each $d$, the dimension of the linear space 
\eqref{eq: highest vector space} equals $m_d^{(n)}$.
When $s \neq 0$,
the spaces~\eqref{eq: highest vector space} themselves are not 
representations of $\Uqsltwo$. 

For $k=1,2,\ldots,2N$,
we will first construct vectors $u_\walk^{(k)} \in \HWsp_{k}^{(\walk(k))}$, 
which can be thought of as being indexed by the first $k$ steps of the walk $\walk$.
From these vectors we will then construct the vectors $\Coblobastwodim_\alpha$ --- see Equation~\eqref{eq: coblobas vectors} below.

Let $u_\walk^{(0)} := 1 \in \bC \isom \Wd_2^{\tens 0}$. 
Define recursively
$u_\walk^{(k+1)} \in \Wd_2^{\tens (k+1)}$ in terms of $u_\walk^{(k)} \in \Wd_2^{\tens k}$ by
\begin{align}\label{eq: coblobas vector recursive construction}
u_\walk^{(k+1)} := \; & \begin{cases}
\Wbas_0 \tens u_\walk^{(k)} & \text{if }\walk(k+1) = \walk(k)+1 \\
\frac{1}{q-q^{-1}}\left( \Wbas_1 \tens u_\walk^{(k)} 
- \frac{q^{\walk(k)}}{\qnum{\walk(k)}}\,\Wbas_0 \tens F.u_\walk^{(k)} \right) 
& \text{if } \walk(k+1) = \walk(k)-1. \end{cases}
\end{align}

\begin{lem}\label{lem: recursively defined vectors}
For $k=0,1,\ldots,2N$,
the vectors $u_\walk^{(k)} \in \Wd_2^{\tens k}$
satisfy $u_\walk^{(k)} \in \HWsp_{k}^{(\walk(k))}$, that is, we have
\begin{align}\label{eq: hwv properties of u} 
E.u_\walk^{(k)} = 0 \qquad \text{and} \qquad K.u_\walk^{(k)} 
= q^{\walk(k)}\,u_\walk^{(k)}.
\end{align}
\end{lem}
\begin{proof}
We prove the assertion by induction on $k$ relying on a direct calculation. 
The base case $k=0$ is clear. Assuming that the claim holds for $u_\walk^{(k)}$,
we verify it for $u_\walk^{(k+1)}$.
Recall that the actions of $E$ and $K$ on $\Wd_2 \tens \Wd_2^{\tens k}$ are given by the 
coproduct~\eqref{eq: coproduct}. We 
use the identities
\begin{align*}
E.\Wbas_0 = 0, \qquad E.\Wbas_1 = \Wbas_0, \qquad
K.\Wbas_0 = q\,\Wbas_0, \qquad K.\Wbas_1 = q^{-1}\,\Wbas_1, \qquad
E.u_\walk^{(k)} = 0, \qquad K.u_\walk^{(k)} = q^{\walk(k)}\,u_\walk^{(k)}.
\end{align*}
If $\walk(k+1) = \walk(k)+1$, we have 
$u_\walk^{(k+1)} = \Wbas_0 \tens u_\walk^{(k)}$ and we easily calculate
\begin{align*}
E.u_\walk^{(k+1)} = \; & E.\Wbas_0 \tens K.u_\walk^{(k)}
 + 1.\Wbas_0 \tens E.u_\walk^{(k)} = 0 \\
K.u_\walk^{(k+1)} = \; & K.\Wbas_0 \tens K.u_\walk^{(k)} 
= q^{1+\walk(k)} \, u_\walk^{(k+1)} = q^{\walk(k+1)}\,u_\walk^{(k+1)}.
\end{align*}

If $\walk(k+1) = \walk(k)-1$, 
then we have
$u_\walk^{(k+1)} 
= \frac{1}{q-q^{-1}}\left( \Wbas_1 \tens u_\walk^{(k)} 
- \frac{q^{\walk(k)}}{\qnum{\walk(k)}}\,\Wbas_0 \tens F.u_\walk^{(k)} \right)$, 
and we similarly~get
\begin{align*}
E.u_\walk^{(k+1)} = \; & \frac{1}{q-q^{-1}} \left(
E.\Wbas_1 \tens K.u_\walk^{(k)} 
- q^{\walk(k)}\, \frac{E.\Wbas_0 \tens KF.u_\walk^{(k)}}{\qnum{\walk(k)}}
+ 1.\Wbas_1 \tens E.u_\walk^{(k)} 
- q^{\walk(k)}\,
\frac{1.\Wbas_0 \tens EF.u_\walk^{(k)}}{\qnum{\walk(k)}} \right) \\
= \; & \frac{1}{q-q^{-1}}\left(
E.\Wbas_1 \tens K.u_\walk^{(k)} 
- \frac{q^{\walk(k)}}{\qnum{\walk(k)}}\,\Wbas_0 \tens 
\frac{(K-K^{-1}).u_\walk^{(k)}}{q-q^{-1}} \right) \\
= \; & \frac{1}{q-q^{-1}}\left(
q^{\walk(k)}\, \Wbas_0 \tens u_\walk^{(k)} 
- \frac{q^{\walk(k)}}{\qnum{\walk(k)}} \,\Wbas_0 \tens 
\frac{q^{\walk(k)}-q^{-\walk(k)}}{q-q^{-1}} \, u_\walk^{(k)} \right) \\
= \; & \frac{q^{\walk(k)}}{q-q^{-1}}\left(
\Wbas_0 \tens u_\walk^{(k)} 
- \frac{\qnum{\walk(k)}}{\qnum{\walk(k)}} \,\Wbas_0 \tens u_\walk^{(k)} \right) 
= 0,
\end{align*}
where we also used the commutation relation 
$EF-FE = \frac{1}{q-q^{-1}}\left(K-K^{-1}\right)$ from~\eqref{eq: quantum group relations}.

Finally, using the commutation relation $KF = q^{-2} FK$  from~\eqref{eq: quantum group relations}, we get
(still with $\walk(k+1) = \walk(k)-1$)
\begin{align*}
K.u_\walk^{(k+1)} = \; & \frac{1}{q-q^{-1}}\left(
K.\Wbas_1 \tens K.u_\walk^{(k)} 
- q^{\walk(k)}\, 
\frac{K.\Wbas_0 \tens KF.u_\walk^{(k)}}{\qnum{\walk(k)}} \right) \\
= \; & \frac{1}{q-q^{-1}}\left(
q^{-1+\walk(k)}\, \Wbas_1 \tens u_\walk^{(k)} 
- q^{1-2+2\walk(k)} \, \frac{\Wbas_0 \tens F.u_\walk^{(k)}}{\qnum{\walk(k)}} 
\right) \\
= \; & q^{\walk(k)-1}\,u_\walk^{(k+1)} = q^{\walk(k+1)}\,u_\walk^{(k+1)}.
\end{align*}
This concludes the proof.
\end{proof}

The vectors $\Coblobastwodim_\alpha$ corresponding to 
the conformal block functions $\ConfBlockFun_\alpha$
are obtained by taking the last
of the recursively defined vectors above, $u_\walk^{(2N)}$, and normalizing it appropriately.
Specifically, for 
$\alpha \in \DP_N$, we set
\begin{align}\label{eq: coblobas vectors}
\Coblobastwodim_\alpha 
:= \qnum{2}^N \NormalizationConstant_\alpha \times u_{\alpha}^{(2N)},
\qquad \text{ where } \qquad
\NormalizationConstant_\alpha:=
\Big( \prod_{\upwedgeat{i} \in \alpha} \frac{1}{\qnum{\alpha(i)+1}} \Big)
\Big( \prod_{\downwedgeat{i} \in \alpha} \qnum{\alpha(i)+1} \Big) .
\end{align}
We finish this subsection by noting that these vectors indeed belong
to the trivial subrepresentation~\eqref{eq: trivial submodule}.

\begin{cor} \label{cor: coblo walk hwv properties}
We have $\Coblobastwodim_\alpha \in \HWsp_{2N}^{(0)}$ 
for all $\alpha \in \DP_N$.
\end{cor}
\begin{proof}
This follows immediately from the properties
\eqref{eq: hwv properties of u} of $u_{\alpha}^{(k)}$ with $k = 2N$.
\end{proof}

\subsubsection{\label{sss: linear independence}\textbf{Linear independence}}

We now quickly verify the linear independence of the
vectors $\Coblobastwodim_\alpha$ constructed in~\eqref{eq: coblobas vector recursive construction}
and~\eqref{eq: coblobas vectors}.
Since we have $\dmn \HWsp_{2N}^{(0)} = \Catalan_N = \# \DP_N$,
linear independence also implies that the collection
$(\Coblobastwodim_\alpha)_{\alpha \in \DP_N}$ is a basis of $\HWsp_{2N}^{(0)}$.

By the recursive construction~\eqref{eq: coblobas vector recursive construction},
the first $k$ steps of $\alpha$ determine a vector $u_\alpha^{(k)} \in \Wd_2^{\tens k}$.
Inductively on $k$, it is clear that all different initial segments of $k$ steps
define linearly independent vectors. The linear independence of
$(\Coblobastwodim_\alpha)_{\alpha \in \DP_N}$ follows from the case $k=2N$.

\subsubsection{\textbf{Projection properties}}

To prove the projection properties~\eqref{eq: coblo walk projection conditions}
for the vectors $\Coblobastwodim_\alpha$ 
constructed in~\eqref{eq: coblobas vector recursive construction}
and~\eqref{eq: coblobas vectors},
we use a recursion property of the normalization coefficients
$\NormalizationConstant_\alpha$. 

\begin{lem}\label{lem: recursion for normalization constants}
The coefficients $\NormalizationConstant_\alpha$ satisfy the 
following recursion: for any $j$, we have
\begin{align}\label{eq: recursion for coblobas normalization constants}
\NormalizationConstant_\alpha= \begin{cases}
\frac{\qnum{\alpha(j)}}{\qnum{\alpha(j)+1}} \times
 \NormalizationConstant_{\alpha \removeupwedge{j}} & \text{if } \upwedgeat{j} \in \alpha \\ 
\frac{\qnum{\alpha(j)+1}}{\qnum{\alpha(j)+2}}\times
 \NormalizationConstant_{\alpha \removedownwedge{j}} & \text{if } \downwedgeat{j} \in \alpha . 
\end{cases}
\end{align}
\end{lem}
\begin{proof}
Observe that the coefficients in~\eqref{eq: coblobas vectors} can be written in the form
\begin{align*}
\prod_{i=1}^{2N} \frac{\sqrt{\qnum{\min\{\alpha(i-1),\alpha(i)\} + 1}}}{\sqrt{\qnum{\max\{\alpha(i-1),\alpha(i)\} + 1}}}
= \Big( \prod_{\upwedgeat{i} \in \alpha} \frac{1}{\qnum{\alpha(i)+1}} \Big)
\Big( \prod_{\downwedgeat{i} \in \alpha} \qnum{\alpha(i)+1} \Big) = \NormalizationConstant_\alpha.
\end{align*}
The expression on the left clearly satisfies 
the recursion~\eqref{eq: recursion for coblobas normalization constants}.
\end{proof}

We 
also make use of the following explicit formulas for the projection $\hat{\pi}$ defined in Equation~\eqref{eq: projection to singlet}.

\begin{lem}
\label{lem: projection formulas}
With $\Sbas \in \Wd_1$ defined in~\eqref{eq: singlet basis vector},
we have $\pi(v)=\hat{\pi}(v) \, \Sbas$ for any $v\in\Wd_{2}\tens\Wd_{2}$, and
\begin{align*}
&\hat{\pi}(\Wbas_{0}\tens\Wbas_{0})=0,\qquad\qquad\quad\qquad\qquad\hat{\pi}(\Wbas_{1}\tens\Wbas_{1})=0,\\
&\hat{\pi}(\Wbas_{0}\tens\Wbas_{1})=\frac{q^{-1}-q}{\qnum 2},\qquad\qquad\qquad\hat{\pi}(\Wbas_{1}\tens\Wbas_{0})=\frac{1-q^{-2}}{\qnum 2}.
\end{align*}
\end{lem}
\begin{proof}
See, e.g.,~\cite[Lemma~2.3]{KP-pure_partition_functions_of_multiple_SLEs}.
\end{proof}


\begin{prop}\label{prop: coblo walk projection properties}
The vectors $( \Coblobastwodim_\alpha )_{\alpha \in \DP}$, 
defined in~\eqref{eq: coblobas vectors}, 
satisfy the projection properties~\eqref{eq: coblo walk projection conditions}.
\end{prop}
\begin{proof}
Fix $j\in\set{1,\ldots,2N-1}$. As the projection $\pi_j$ acts locally
on the $j$:th and $(j+1)$:st tensor components, the value of 
$\hat{\pi}_j(\Coblobastwodim_\alpha)$ can be calculated using 
the explicit construction~\eqref{eq: coblobas vector recursive construction} 
and the 
recursion~\eqref{eq: recursion for coblobas normalization constants} of 
Lemma~\ref{lem: recursion for normalization constants} for the normalization 
constants appearing in the definition~\eqref{eq: coblobas vectors} 
of~$\Coblobastwodim_\alpha$.
We treat separately each possible local shape of a Dyck path $\alpha$
at $j$, i.e., the cases depicted in Figure~\ref{fig: wedges and slopes}.

Suppose first that $\alpha$ contains a slope 
at $j$, i.e., $\slopeat{j} \in \alpha$.
We need to show that in this case, we have 
$\hat{\pi}_j(\Coblobastwodim_\alpha) = 0$, or, equivalently, that
$\hat{\pi}_j(u_{\alpha}^{(j+1)}) = 0$. 
Depending whether the slope is an up-slope or a down-slope,
we study the two cases in~\eqref{eq: coblobas vector recursive construction}.

In the easiest case of an up-slope, that is, when we have $\alpha(j) = \alpha(j-1)+1$ and 
$\alpha(j+1) = \alpha(j)+1$, the tensor components $j$ and $j+1$ in 
$u_{\alpha}^{(j+1)}$ (counting from the right)
are proportional to $\Wbas_0 \tens \Wbas_0$,
and $\hat{\pi}_j$ thus annihilates the vector $u_{\alpha}^{(j+1)}$ 
by Lemma~\ref{lem: projection formulas}(a). 
Equations~\eqref{eq: coblobas vector recursive construction} 
and~\eqref{eq: coblobas vectors} then show that we also have
$\hat{\pi}_j(\Coblobastwodim_\alpha) = 0$, as asserted in~\eqref{eq: coblo walk projection conditions}.

In the case of a down-slope, that is, when we have $\alpha(j) = \alpha(j-1)-1$ and 
$\alpha(j+1) = \alpha(j)-1$, the tensor components $j$ and $j+1$ in 
$u_{\alpha}^{(j+1)}$ have several terms. To perform the calculations,
it is convenient to first write down the action of $F$ on 
$u_{\alpha}^{(j)}$.
The action is given by the coproduct~\eqref{eq: coproduct} as follows:
\begin{align*}
(q-q^{-1}) \, F.u_\alpha^{(j)} \; & = 
F. \left( \Wbas_1 \tens u_\alpha^{(j-1)} 
- \frac{q^{\alpha(j-1)}}{\qnum{\alpha(j-1)}}\,\Wbas_0 \tens F.u_\alpha^{(j-1)} \right) \\
\; & = F.\Wbas_1 \tens 1.u_\alpha^{(j-1)} 
+ K^{-1}.\Wbas_1 \tens F.u_\alpha^{(j-1)}
- q^{\alpha(j-1)}\, \frac{F.\Wbas_0 \tens F.u_\alpha^{(j-1)}
- K^{-1}.\Wbas_0 \tens F^2.u_\alpha^{(j-1)}}{\qnum{\alpha(j-1)}} \\
\; & = q\, \Wbas_1 \tens F.u_\alpha^{(j-1)} 
- \frac{q^{\alpha(j-1)}}{\qnum{\alpha(j-1)}}\left( \Wbas_1 \tens F.u_\alpha^{(j-1)}
- q^{-1} \, \Wbas_0 \tens F^2.u_\alpha^{(j-1)} \right) \\
\; & = \left( q - \frac{q^{\alpha(j-1)}}{\qnum{\alpha(j-1)}} \right) 
\Wbas_1 \tens F.u_\alpha^{(j-1)}
- \frac{q^{\alpha(j-1)-1}}{\qnum{\alpha(j-1)}}\, \Wbas_0 \tens F^2.u_\alpha^{(j-1)},
\end{align*}
where we used the identities $F.\Wbas_1 = 0$,
$F.\Wbas_0 = \Wbas_1$, $K^{-1}.\Wbas_1 = q\,\Wbas_1$, 
and $K^{-1}.\Wbas_0 = q^{-1}\,\Wbas_0$.
The vector $u_\alpha^{(j+1)}$ now reads
\begin{align*}
u_\alpha^{(j+1)} \propto \; &
\Wbas_1 \tens u_\alpha^{(j)} 
- \frac{q^{\alpha(j)}}{\qnum{\alpha(j)}}\,\Wbas_0 \tens F.u_\alpha^{(j)} \\
\propto \; &  \Wbas_1 \tens \left(
\Wbas_1 \tens u_\alpha^{(j-1)} 
- \frac{q^{\alpha(j-1)}}{\qnum{\alpha(j-1)}}\,\Wbas_0 \tens F.u_\alpha^{(j-1)} \right)  \\
\; & - \frac{q^{\alpha(j)}}{\qnum{\alpha(j)}}\,\Wbas_0 \tens \left(
\left( q - \frac{q^{\alpha(j-1)}}{\qnum{\alpha(j-1)}} \right) 
\Wbas_1 \tens F.u_\alpha^{(j-1)}
- \frac{q^{\alpha(j-1)-1}}{\qnum{\alpha(j-1)}}\, \Wbas_0 \tens F^2.u_\alpha^{(j-1)} \right).
\end{align*}

Using Lemma~\ref{lem: projection formulas}(a), the down-step
$\alpha(j) = \alpha(j-1)-1$, and the geometric sum expansion
of the $q$-integers $\qnum{n} = q^{n-1} + q^{n-3} + \cdots + q^{3-n} + q^{1-n}$, 
we verify that
\begin{align*}
\hat{\pi}_j(u_\alpha^{(j+1)})
\propto \; &  
\left(
- \frac{q^{\alpha(j-1)}}{\qnum{\alpha(j-1)}}\,
\hat{\pi} ( \Wbas_1 \tens \Wbas_0 )
- \frac{q^{\alpha(j)+1} - \frac{q^{\alpha(j)+\alpha(j-1)}}{\qnum{\alpha(j-1)}}}{\qnum{\alpha(j)}}\,
\hat{\pi} ( \Wbas_0 \tens \Wbas_1 ) 
\right) \tens F.u_\alpha^{(j-1)} \\
= \; & \left(
- \frac{q^{\alpha(j)+1}}{\qnum{\alpha(j-1)}}\,
\frac{1-q^{-2}}{\qnum 2}
- \frac{q^{\alpha(j)+1} - \frac{q^{2\alpha(j)+1}}{\qnum{\alpha(j-1)}}}{\qnum{\alpha(j)}}\, \frac{q^{-1}-q}{\qnum 2} \right) \tens F.u_\alpha^{(j-1)} \\
= \; & \frac{q^{\alpha(j)+1}(q-q^{-1})}{\qnum 2 \qnum{\alpha(j-1)} \qnum{\alpha(j)}} \times 
\left( - q^{-1}\,\qnum{\alpha(j)}
+ \qnum{\alpha(j)+1} - q^{\alpha(j)} \right) \tens F.u_\alpha^{(j-1)}  \\
= \; & 0.
\end{align*}
It thus follows by 
Equations~\eqref{eq: coblobas vector recursive construction} 
and~\eqref{eq: coblobas vectors} 
that the asserted property $\hat{\pi}_j(\Coblobastwodim_\alpha) = 0$ holds
also with $\alpha$ having an down-slope at $j$.

Suppose then that $\alpha$ contains an up-wedge 
at $j$, i.e., $\upwedgeat{j} \in \alpha$.
We need to show that in this case, we have 
$\hat{\pi}_j(\Coblobastwodim_\alpha) 
= \Coblobastwodim_{\alpha \removeupwedge{j}}$. 
Now $\alpha(j) = \alpha(j-1)+1$ and $\alpha(j+1) = \alpha(j)-1$
and the vector $u_\alpha^{(j+1)}$ reads
\begin{align*}
u_\alpha^{(j+1)} 
= \; & \frac{1}{q-q^{-1}}\left( 
\Wbas_1 \tens (\Wbas_0 \tens u_\alpha^{(j-1)})
- \frac{q^{\alpha(j)}}{\qnum{\alpha(j)}}\,
\Wbas_0 \tens F.(\Wbas_0 \tens u_\alpha^{(j-1)}) \right) \\
= \; & \frac{1}{q-q^{-1}}\left(
\Wbas_1 \tens (\Wbas_0 \tens u_\alpha^{(j-1)})
- \frac{q^{\alpha(j)}}{\qnum{\alpha(j)}}\,\Wbas_0 \tens 
(F.\Wbas_0 \tens 1.u_\alpha^{(j-1)} + K^{-1}.\Wbas_0 \tens F.u_\alpha^{(j-1)}) \right) \\
= \; & \frac{1}{q-q^{-1}}\left(
\Wbas_1 \tens (\Wbas_0 \tens u_\alpha^{(j-1)}) 
- \frac{q^{\alpha(j)}}{\qnum{\alpha(j)}}\,\Wbas_0 \tens (\Wbas_1 \tens u_\alpha^{(j-1)} + q^{-1}\, \Wbas_0 \tens F.u_\alpha^{(j-1)}) \right).
\end{align*}

Applying the projection $\hat{\pi}_j$ on both sides and using
Lemma~\ref{lem: projection formulas}(a), we obtain
\begin{align*}
\hat{\pi}_j ( u_\alpha^{(j+1)} )
= \; & \frac{1}{q-q^{-1}}\left(
\hat{\pi} ( \Wbas_1 \tens \Wbas_0)
- \frac{q^{\alpha(j)}}{\qnum{\alpha(j)}}\, 
\hat{\pi} ( \Wbas_0 \tens \Wbas_1) \right) \tens u_\alpha^{(j-1)} \\
= \; & \frac{1}{q-q^{-1}} \left(
\frac{1-q^{-2}}{\qnum 2}
- \frac{q^{\alpha(j)}}{\qnum{\alpha(j)}}\,\frac{q^{-1}-q}{\qnum 2}
\right) \times u_\alpha^{(j-1)} \\
= \; & \frac{1}{\qnum 2 \qnum{\alpha(j)}} 
\left( q^{-1}\,\qnum{\alpha(j)} + q^{\alpha(j)} \right)
\times u_\alpha^{(j-1)}.
\end{align*}
Using again the geometric sum expansion
of the $q$-integers, 
we simplify
the multiplicative factor by
$q^{-1}\,\qnum{\alpha(j)} + q^{\alpha(j)} =  \qnum{\alpha(j)+1}$,
which yields
\begin{align*}
\hat{\pi}_j ( u_\alpha^{(j+1)} )
= \; & \frac{\qnum{\alpha(j)+1}}{\qnum 2 \qnum{\alpha(j)}} \times u_\alpha^{(j-1)} .
\end{align*}
By Equations~\eqref{eq: coblobas vector recursive construction} 
and~\eqref{eq: coblobas vectors} and the 
recursion~\eqref{eq: recursion for coblobas normalization constants}, the 
asserted property~\eqref{eq: coblo walk projection conditions} follows:
\begin{align*}
\hat{\pi}_j ( \Coblobastwodim_\alpha ) 
= \; & \qnum{2}^N \NormalizationConstant_\alpha \times 
\hat{\pi}_j ( u_{\alpha}^{(2N)} ) \\
= \; & \qnum{2}^N \frac{\qnum{\alpha(j)}}{\qnum{\alpha(j)+1}} \times
\NormalizationConstant_{\alpha\removeupwedge{j}} \times 
\frac{\qnum{\alpha(j)+1}}{\qnum 2 \qnum{\alpha(j)}}
\times u_{\alpha \removeupwedge{j}}^{(2N-2)} \\
= \; & \qnum{2}^{N-1} \NormalizationConstant_{\alpha\removeupwedge{j}} \times
u_{\alpha \removeupwedge{j}}^{(2N-2)} \\
= \; & \Coblobastwodim_{\alpha \removeupwedge{j}}.
\end{align*}

Finally, suppose that $\alpha$ contains a down-wedge 
at $j$, i.e., $\downwedgeat{j} \in \alpha$.
We need to show that in this case, we have 
$\hat{\pi}_j(\Coblobastwodim_\alpha) =
-\frac{\qnum{\alpha(j)+1}}{\qnum{\alpha(j)+2}} \times \Coblobastwodim_{\alpha\removedownwedge{j}}$. 
Now $\alpha(j) = \alpha(j-1)-1$ and $\alpha(j+1) = \alpha(j)+1$
and $u_\alpha^{(j+1)}$ reads
\begin{align*}
u_\alpha^{(j+1)} 
= \; & \frac{1}{q-q^{-1}} \left( 
\Wbas_0 \tens (\Wbas_1 \tens u_\alpha^{(j-1)})
- \frac{q^{\alpha(j-1)}}{\qnum{\alpha(j-1)}}\,\Wbas_0 \tens (\Wbas_0 \tens F.u_\alpha^{(j-1)}) 
\right).
\end{align*}
Applying the projection $\hat{\pi}_j$ on both sides and using
Lemma~\ref{lem: projection formulas}(a), we obtain
\begin{align*}
\hat{\pi}_j ( u_\alpha^{(j+1)} )
= \; & \frac{1}{q-q^{-1}} \left( 
\hat{\pi} (\Wbas_0 \tens \Wbas_1) \right) \tens u_\alpha^{(j-1)} 
= \frac{1}{q-q^{-1}} \frac{q^{-1}-q}{\qnum 2} \times u_\alpha^{(j-1)}
= - \frac{1}{\qnum 2} \times u_\alpha^{(j-1)},
\end{align*}
and again, by Equations~\eqref{eq: coblobas vector recursive construction} 
and~\eqref{eq: coblobas vectors} 
and the 
recursion~\eqref{eq: recursion for coblobas normalization constants}, the asserted
property~\eqref{eq: coblo walk projection conditions} follows:
\begin{align*}
\hat{\pi}_j ( \Coblobastwodim_\alpha ) 
= \; & \qnum{2}^N \NormalizationConstant_\alpha \times 
\hat{\pi}_j ( u_{\alpha}^{(2N)} ) \\
= \; & \qnum{2}^N \frac{\qnum{\alpha(j)+1}}{\qnum{\alpha(j)+2}}  \times
\NormalizationConstant_{\alpha\removedownwedge{j}} \times 
- \frac{1}{\qnum 2} 
\times u_{\alpha \removedownwedge{j}}^{(2N-2)} \\
= \; & - \frac{\qnum{\alpha(j)+1}}{\qnum{\alpha(j)+2}} 
\qnum{2}^{N-1} \NormalizationConstant_{\alpha\removedownwedge{j}} \times 
u_{\alpha \removedownwedge{j}}^{(2N-2)} \\
= \; & - \frac{\qnum{\alpha(j)+1}}{\qnum{\alpha(j)+2}} 
\times \Coblobastwodim_{\alpha \removedownwedge{j}}.
\end{align*}
This concludes the proof.
\end{proof}

\subsubsection{\textbf{Proof of Proposition~\ref{prop: conformal block vectors}}}

The vectors $(\Coblobastwodim_\alpha)_{\alpha \in \DP}$ constructed
in~\eqref{eq: coblobas vector recursive construction} and~\eqref{eq: coblobas vectors}
lie in the space $\HWsp_{2N}^{(0)}$ by Corollary~\ref{cor: coblo walk hwv properties}
and satisfy the projection properties by Proposition~\ref{prop: coblo walk projection properties}.
Such a collection is unique by Corollary~\ref{cor: uniqueness}.
In Section~\ref{sss: linear independence} we verified that
$(\Coblobastwodim_\alpha)_{\alpha \in \DP_N}$ forms a basis of $\HWsp_{2N}^{(0)}$.
$\hfill \qed$

\bibliographystyle{annotate}

\newcommand{\etalchar}[1]{$^{#1}$}

\end{document}